%% file: main.tex
\documentclass[a4paper,UKenglish,cleveref, autoref, thm-restate, final]{lipics-v2021}

\usepackage{stackrel}
\usepackage{extarrows}

\usepackage{thm-restate}
\usepackage{xcolor}
\usepackage{pict2e} 
\usepackage{tikz}
\usepackage[obeyFinal]{todonotes}
\usepackage{xspace}
\usepackage{wrapfig}
\usetikzlibrary{automata,positioning,arrows}
\tikzset{
->, 
>=stealth, 
node distance=3cm, 
initial text=$$ 
}
\usetikzlibrary{arrows,automata}
\usetikzlibrary{shapes.multipart}
\graphicspath{{figures/}{../}}

\input{macros}

\title{A robust class of languages of 2-nested words} 

\author{Séverine Fratani}{LIS, Aix-Marseille Univ, CNRS, Marseille, France}{}{}{}
\author{Guillaume Maurras}{LIS, Aix-Marseille Univ, CNRS, Marseille, France}{}{}{}
\author{Pierre-Alain Reynier}{LIS, Aix-Marseille Univ, CNRS, Marseille, France}{}{}{}

\authorrunning{S. Fratani, G. Maurras and P.-A. Reynier} 

\Copyright{Jane Open Access and Joan R. Public} 

\ccsdesc{Theory of computation~Models of computation}
\ccsdesc{Theory of computation~Formal languages and automata theory}

\keywords{Nested word, Determinization, Indexed languages} 

\nolinenumbers 

\EventEditors{John Q. Open and Joan R. Access}
\EventNoEds{2}
\EventLongTitle{42nd Conference on Very Important Topics (CVIT 2016)}
\EventShortTitle{CVIT 2016}
\EventAcronym{CVIT}
\EventYear{2016}
\EventDate{December 24--27, 2016}
\EventLocation{Little Whinging, United Kingdom}
\EventLogo{}
\SeriesVolume{42}
\ArticleNo{23}

\begin{document}

\maketitle

\begin{abstract}
Regular nested word languages (a.k.a. visibly pushdown languages)
strictly extend regular word languages, while preserving
their main closure and decidability properties. Previous works have shown that 
considering languages of 2-nested words, \emph{i.e.} words enriched with two matchings 
(a.k.a. $2$-visibly pushdown languages), 
is not as successful: the corresponding model of automata is not closed
under determinization.
In this work, inspired by homomorphic representations of indexed languages,
we identify a subclass of $2$-nested words, which we call $2$-wave words.
This class strictly extends the class of nested words, while preserving its main properties.
More precisely, we prove closure under determinization of the corresponding automaton model, 
we provide a logical characterization of the recognized languages,
and show that the corresponding graphs have bounded treewidth. As a consequence, we derive important
closure and decidability properties. Last, we show that the word projections of the languages
we define belong to the class of linear indexed languages.
\end{abstract}

\section{Introduction}
\label{sec:introduction}

\input{intro}

\section{Words and matchings}
\label{sec:preliminaries}

\input{prelim}

\section{2-Nested Word Automata}
\label{sec:automata}

\input{automata}

\section{Determinization of 2NWA over 2-wave words}
\label{sec:determinization}

\input{determinization}

\section{Logical characterization}
\label{sec:applications}

\input{applications}
\section{Decision problems}
\label{sec:decision}

\input{decision2}

\section{Relation to indexed languages}
\label{sec:discussion}

\input{discussion}

\section{Conclusion}
\label{sec:conclusion}

A natural perspective of this work consists in studying
its extension to $k$-wave words. We believe that the
determinization property should hold for this class too.
This will require to improve
our arguments, as the present proof seems intricate to be
adapted to $k$-waves. 

For unbounded waves, automata are not determinizable: Bollig proved in \cite{DBLP:journals/lmcs/Bollig08} that regular \tnwl are not closed under complementation, and are then not determinizable.   
His proof uses
the encoding of grids in \tnwl, and can thus be adapted to (unbounded) wave words.

Another perspective consists in characterizing the subclass of  indexed languages capturing languages $\exists \ww L$, when $L$ is a regular $\ww_2$ language, and determine if   $\exists \ww L$ is still indexed when $L$  is a regular $\ww_k$ language,  for $k>2$.

%
%


\input{main.bbl}
\newpage
\appendix

\input{app-grammar}

\input{app-closure}

\input{app-deter}
\input{app-logic}

\input{app-tw}

\input{app-discussion}

\end{document}

%% file: macros.tex
\newcommand{\gui}[1]{}
\newcommand{\pa}[1]{}

\DeclareFontFamily{U}{mathx}{\hyphenchar\font45}
\DeclareFontShape{U}{mathx}{m}{n}{<-> mathx10}{}
\DeclareSymbolFont{mathx}{U}{mathx}{m}{n}
\DeclareMathAccent{\widebar}{0}{mathx}{"73}

\newcommand{\CFL}{CFL\xspace}
\newcommand{\ww}{\textup{\textsf{WW}}\xspace}
\newcommand{\nw}{\mathsf{NW}\xspace}
\newcommand{\nwa}{\textup{\textsf{NWA}}\xspace}
\newcommand{\nwl}{\textup{\textsf{NWL}}\xspace}
\newcommand{\tnw}{\mathsf{2NW}\xspace}
\newcommand{\tnwl}{\textup{\textsf{2NWL}}\xspace}
\newcommand{\tnwa}{\textup{\textsf{2NWA}}\xspace}
\newcommand{\il}{\textup{\textsf{IL}}\xspace}
\newcommand{\lil}{\textup{\textsf{LIL}}\xspace}
\newcommand{\existsw}{\textup{\textsf{$\exists$W}}\xspace}

\newcommand{\trans}[2]{\Delta^{#1}_{#2}}
\newcommand{\transd}[3]{{#1}^{#2}_{#3}}

\newcommand{\Ientff}[2]{[\![ #1,#2 ]\!]}
\newcommand{\Ientof}[2]{]\!] #1,#2 ]\!]}
\newcommand{\Ientoo}[2]{]\!] #1,#2 [\![}
\newcommand{\Ientfo}[2]{[\![ #1,#2 [\![}
\newcommand{\interv}[1]{[\![#1 ]\!]}

\newcommand{\set}[1]{\{ #1 \}}
\newcommand{\eps}{\varepsilon}
\newcommand{\dyck}{\mathcal{D}}

\newcommand{\cpl}{M}
\newcommand{\match}[1]{\cpl(#1)}
\newcommand{\imatch}[1]{\cpl^{-1}(#1)}

\newcommand{\cpli}{\cpl_1}
\newcommand{\cplii}{\cpl_2}

\newcommand{\imatchun}[1]{\cpli^{-1}(#1)}
\newcommand{\imatchdeux}[1]{\cplii^{-1}(#1)}

\newcommand{\qqs}{\forall}
\newcommand{\ilx}{\exists}
\newcommand{\non}{\lnot}

\newcommand{\et}{\wedge}
\newcommand{\imp}{\Longrightarrow}
\newcommand{\ssi}{\Longleftrightarrow}

\newcommand{\exec}[4]{#1\xlongrightarrow[#4]{#2}#3}



%% file: intro.tex


The class of regular languages constitutes a cornerstone of theoretical computer
science, thanks to its numerous closure and decidability properties.
A long line of research  studied extensions of this class, while preserving
its robustness. 
Context free languages (\CFL for short) constitute a very important class: they admit multiple presentations,
by means of pushdown automata, context-free grammars and more, and have 
led to numerous applications. Unfortunately, \CFL do not satisfy several 
important properties enjoyed by regular languages. More precisely, the corresponding
automaton model, namely pushdown automata, does not admit determinization.
In addition, the class of \CFL is not closed under intersection nor complement,
and universality, inclusion and equivalence are undecidable properties.

A simple way to patch this is by considering as input the word together 
with the inherent matching relation, resulting in what is known as a \emph{nested 
word}~\cite{DBLP:journals/jacm/AlurM09}.
Indeed, as soon as a word belongs to
a \CFL, one can identify a matching relation on (some of) the positions of the word,
whatever the presentation of the \CFL. For instance, if the \CFL is given as a pushdown automaton,
then this relation associates push/pop positions. Another way to define
the matching relation is to use an alternative presentation of \CFL
given in~\cite{conf/dlt/Okhotin12}, which refines the Chomsky-Sch\"utzenberger theorem.
Following this work, one can show that a language $L$ is a \CFL
iff there exists a regular language $R$, a Dyck language $D_2$ over two pairs of brackets, and two homomorphisms $h$ ($h$ is non-erasing), $g$
such that $L = h(g^{-1}(D_2)\cap R)$. This alternative presentation also leads to
a natural matching relation, induced by $D_2$. 
The model of \emph{nested word automata} naturally extends finite-state automata by allowing 
to label edges of the matching relation with states (often called hierarchical). This model
accepts the so-called class of \emph{regular nested word languages},
allowing to recover most of the
nice properties
of regular languages. 
More precisely, nested word automata can always be determinized. The class of regular nested word languages is closed under all the boolean operations, 
admits an equivalent presentation by means of logic (monadic-second order logic with a binary predicate corresponding to the matching relation), and expected
decidability properties (emptiness, universality, inclusion and equivalence).
It is worth noticing that another way to present this class is by splitting the alphabet into call/return/internal symbols. This leads to so-called \emph{visibly pushdown languages}~\cite{DBLP:conf/stoc/AlurM04}, and the associated model
of visibly pushdown automata. 

Several works tried to extend the class of regular nested word languages 
while preserving its closure and decidability properties. In particular, a focus has been put
on words with
multiple matching relations.
In~\cite{DBLP:conf/lics/TorreMP07}, the authors consider multiple stacks 
with 
a semantical restriction of push/pop operations known as $k$ phases. 
That way, they obtain
decidability of the emptiness problem. However, their model of 
automata cannot be determinized.
In~\cite{DBLP:journals/tcs/CarotenutoMP16}, the authors consider
visibly pushdown automata with multiple stacks, with an ordering
on stacks, and prove the closure under complement of their model.
Their proof does not rely on the determinization of the model:
indeed, as shown in~\cite{DBLP:conf/lics/TorreMP07}, this class cannot be determinized in general. 
This corrects a previous result published in~\cite{CMP07}, which states that
the general class of 2-stack
visibly pushdown automata is closed under complementation, which  
does not hold, as shown in~\cite{DBLP:journals/lmcs/Bollig08}.
The crux in their proof was the use of determinization of 2-stack
visibly pushdown automata.
In~\cite{DBLP:journals/lmcs/Bollig08}, Bollig studies 2-stack visibly pushdown automata
in their unrestricted form. He proves the equivalence with the existential fragment of monadic second-order logic, but that quantifier alternation leads to an infinite hierarchy in this
setting. As a corollary, the resulting class of languages is not closed under complementation.
Another restriction, known as scope-bounded pushdown languages, has been introduced
in~\cite{TorreMP10,TorreNP16}, for which the authors manage
to prove that the automaton model can be determinized.

The previous survey of related works 
shows the difficulty in identifying a class of $2$-nested words
for which the corresponding automaton model can be determinized. 
As a consequence, the decidability results presented in these papers require
ad-hoc involved proofs. 
Graphs of bounded treewidth constitute an alternative approach for
obtaining decidability properties. 
Indeed, in~\cite{MP11}, the authors show that most of the previous classes 
with good decidability properties
actually correspond to graphs of bounded treewidth, for which
MSO decidability follows from \cite{Courcelle97,Seese91}.
Yet, determinization of nested word automata is the keystone of the
nice properties of this model, and thus constitutes a highly desirable feature.
In the present
work, taking inspiration in indexed languages, we identify a class of $2$-nested words for which automata can be
determinized.
Our class is incomparable with those of~\cite{DBLP:conf/lics/TorreMP07}, \cite{DBLP:journals/tcs/CarotenutoMP16} and~\cite{TorreMP10,TorreNP16}.
Intuitively, between two matched positions of the first matching, they bound 
the number of switches between matchings, while we do not.
In addition, the proof of determinizability of~\cite{TorreMP10,TorreNP16} is  different from ours, as their proof
is a kind of superviser that uses determinization of~\cite{DBLP:journals/jacm/AlurM09} as a subroutine, while
ours generalizes the construction of~\cite{DBLP:journals/jacm/AlurM09} to two nestings.

Indexed languages ~\cite{DBLP:journals/jacm/Aho68} correspond to the level 2 of the infinite
hierarchy of higher-order languages \cite{Mas74}.
With numerous applications in computational linguistics, they
have been much studied during the seventies and the eighties  \cite{Mas76,Damm82,Engel83}. 
Homomorphic characterizations of \CFL that we presented before
have been extended to (linear) indexed languages in
several works, including~\cite{w88,DBLP:journals/iandc/FrataniV19}.
One of them (see~\cite{DBLP:journals/iandc/FrataniV19}) shows that 
$L$ is an indexed language
iff there exists a regular language $R$, two Dyck languages $D_2$ and $D_k$ over two and $k$ pairs of brackets respectively, and two homomorphisms $h$, $g$
such that $L = h(g^{-1}(D_2)\cap R \cap D_k)$, with some additional conditions on $g$. 
This presentation allows to associate with a word $w\in L$ two matchings, induced by
$D_2$ and $D_k$, which interact in a very particular way. Graphically speaking, this interaction
yields kinds of \emph{waves} (see Figure~\ref{fig:waveintro}). When restricted to linear indexed languages, the length of these
waves is upper bounded by $2$, yielding the structure of $2$-waves that will be of interest to us
in this paper. All these notions will be formally presented in the paper.
We also refer the reader to Section~\ref{sec:discussion}, in which we explore the
relationship between our work and (linear) indexed languages.

\input{fig-intro}

In this paper, we consider $2$-nested words whose 
matchings satisfy the structural restriction of $2$-waves;
we call them $2$-wave words. The main result of this paper
is to show that $2$-nested word automata are closed under
determinization on the class of $2$-wave words. 
While determinization of nested word automata extends the 
well-known powerset construction with a reference state, our construction
is more involved in order to be able to reconcile the labels of the different
arches, and we give a detailed proof of correction.
This allows us to prove the equivalence of automata with monadic-second order
logic on $2$-wave words, as well as with its existential fragment. 
This contrasts with results of~\cite{DBLP:journals/lmcs/Bollig08} which
show that on $2$-nested words, quantifier alternation yields an infinite hierarchy.
We also prove that the graphs associated with $2$-wave words are MSO-definable,
and have bounded treewidth. As a consequence, we obtain 
the following decidability results for the class of $2$-wave words: satisfiability of MSO, 
emptiness of automata (in polynomial time), universality, inclusion and
equivalence of automata (in exponential time).

In Section~\ref{sec:preliminaries}, we introduce the definitions of matchings and nested words, and provide a grammar
for $2$-wave words. In Section~\ref{sec:automata}, we introduce the automaton model, and present
in Section~\ref{sec:determinization} the main result of the paper: the closure under
determinization over $2$-wave words. Applications to logic and decidability 
are presented in Sections~\ref{sec:applications} and~\ref{sec:decision}. Last, a discussion on relations with (linear) 
indexed languages is given in Section~\ref{sec:discussion}. 
Omitted proofs can be found in Appendix.

%% file: fig-intro.tex


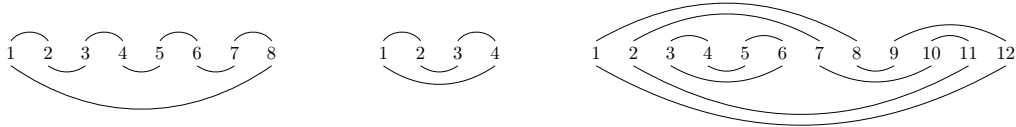
\begin{figure}[h]
\centering
\scalebox{0.7}{
\begin{tikzpicture}
\tikzset{node distance=0.7cm,}
\node (0) {$1$};
\node[right of = 0] (1) {$2$};
\node[right of = 1] (2) {$3$};
\node[right of = 2] (3) {$4$};
\node[right of = 3] (4) {$5$};
\node[right of = 4] (5) {$6$};
\node[right of = 5] (6) {$7$};
\node[right of = 6] (7) {$8$};
\draw[-] (0.north) edge[bend left =55] (1.north);
\draw[-] (2.north) edge[bend left =55] (3.north);
\draw[-] (4.north) edge[bend left =55] (5.north);
\draw[-] (6.north) edge[bend left =55] (7.north);

\draw[-] (1.south) edge[bend right =35] (2.south);
\draw[-] (3.south) edge[bend right =35] (4.south);
\draw[-] (5.south) edge[bend right =35] (6.south);

\draw[-] (0.south) edge[bend right =35] (7.south);

\begin{scope}[shift={(7,0)}]
\node (0) {$1$};
\node[right of = 0] (1) {$2$};
\node[right of = 1] (2) {$3$};
\node[right of = 2] (3) {$4$};
\draw[-] (0.north) edge[bend left =55] (1.north);
\draw[-] (2.north) edge[bend left =55] (3.north);

\draw[-] (1.south) edge[bend right =35] (2.south);

\draw[-] (0.south) edge[bend right =35] (3.south);
\end{scope}

\begin{scope}[shift={(11,0)}]
\node (0) {$1$};
\node[right of = 0] (1) {$2$};
\node[right of = 1] (2) {3};
\node[right of = 2] (3) {4};
\node[right of = 3] (4) {5};
\node[right of = 4] (5) {6};
\node[right of = 5] (6) {7};
\node[right of = 6] (7) {8};
\node[right of = 7] (8) {9};
\node[right of = 8] (9) {10};
\node[right of = 9] (10) {11};
\node[right of = 10] (11) {12};
\draw[-] (0.north) edge[bend left =30] (7.north);
\draw[-] (1.north) edge[bend left =30] (6.north);
\draw[-] (2.north) edge[bend left =30] (3.north);
\draw[-] (4.north) edge[bend left =30] (5.north);
\draw[-] (8.north) edge[bend left =30] (11.north);
\draw[-] (9.north) edge[bend left =30] (10.north);

\draw[-] (0.south) edge[bend right =30] (11.south);
\draw[-] (1.south) edge[bend right =30] (10.south);
\draw[-] (2.south) edge[bend right =30] (5.south);
\draw[-] (3.south) edge[bend right =30] (4.south);
\draw[-] (6.south) edge[bend right =30] (9.south);
\draw[-] (7.south) edge[bend right =30] (8.south);
\end{scope}

\end{tikzpicture}
}
\vspace{-.4cm}
\caption{A $4$-wave (left), a $2$-wave (middle), and a combination of
$2$-waves (right).}
\label{fig:waveintro}
\end{figure}

%% file: prelim.tex

\paragraph*{Words and relations}
For any positive integer $n$, $[n] =\{1,\dots,n\}$ is the set of all positive integers ranging from $1$ to $n$.
When referencing a position in a word, integers might be called \emph{positions} or \emph{index}.
$\Sigma$ denotes a finite alphabet. The empty word is denoted $\epsilon$, and the set of finite
words on $\Sigma$ is denoted $\Sigma^*$. The length of $w\in \Sigma^*$ is denoted $|w|$.
Given a non-empty word $w\in \Sigma^*$, its positions are numbered from $1$ to $|w|$.
\todo{PA: added def factor}
We say that $u$ is a \emph{factor} of $w$ if there exist two words $x,y$
such that $w = xuy$.
Given an interval $I\subseteq [|w|]$, we denote by $w_{|I}$ the factor of $w$
corresponding to positions in $I$.
Unless specified otherwise, words and automata introduced in this paper 
are defined on $\Sigma$.

\begin{definition}  A \emph{matching relation} of length  $n\geq 0$ is a binary 
relation $M$ on $[n]$ such that:
\begin{enumerate}
\item if $\match{i,j}$ then $i<j$, i.e. $\cpl$ is compatible with the natural order on integers
\item if $\match{i,j}$ and $\match{k,l}$ then $\{i, j\} \cap \{k, l\}\neq\emptyset \imp i=k\et j=l$, i.e. any integer is related at most once by the matching
\item if $\match{i,j}$ and $\match{k,l}$ then $i<k<j\imp l<j$ i.e. a matching is non-crossing.
\end{enumerate}
\end{definition}

 Since a matching is an injective functional relation, we will often use functional notations: $\match{i}=j$ or $\imatch{j}=i$ rather than $ \match{i,j}$.
If $\match{i,j}$, we call $i$ a \emph{call} position, $j$ a \emph{return} position and if $k\in [n]$ is neither a call nor a return position we call it an \emph{internal} position. 

If $I$ is a subset of $[n]$ we denote by $I^c$ the subset of call positions of $I$,  and by $I^r$ the subset of return positions and $I^{\text{int}}$ the subset of internal position.
We say that
$I$ is \emph{without pending arch} (wpa) if it has no pending call nor pending return, 
 \emph{i.e.}  $\match{ I^c} \cup  \imatch{ I^r}  \subseteq  I$. Note that this definition holds when $I$ is an interval, but
 even for an arbitrary subset of $[n]$. In particular, 
we say a pair $(I_1,I_2)$ of intervals is without pending arch if $I_1\cup I_2$ is.
In this case, we may also say that $(I_1,I_2)$ is a pair of \emph{matched intervals}.

\paragraph*{(2-)Nested Words}
We first recall the classical definition of nested words:
\begin{definition} A \emph{nested word} on $\Sigma$ is a pair $\omega = (w, \cpl)$ where $w\in\Sigma^*$ and $\cpl$ is a matching of length $|w|$. We write $\nw(\Sigma)$ the set of nested words on $\Sigma$, and
$\nwl(\Sigma)$ the set of languages of nested words. 
\end{definition}

Words equipped with two matchings are called $2$-nested words:
\begin{definition} A \emph{2-nested word} on $\Sigma$ is a triple $\omega = (w, \cpli, \cplii)$ where $w\in\Sigma^*$ and $\cpli$, $\cplii$ are matchings of length $|w|$. We write $\tnw(\Sigma)$ for the set of 2-nested words on $\Sigma$, and
$\tnwl(\Sigma)$ for the set of languages of 2-nested words. 
\end{definition}

\begin{example}
An example of a nested word (resp. two examples of $2$-nested words)
is depicted on the left (resp. on the middle and right)
of Figure~\ref{fig:2nw}. 
For 2-nested words, 
the matching  $\cpli$
is depicted above, while matching $\cplii$ is depicted below.
\end{example}

Given a $2$-nested word $\omega = (w,\cpli,\cplii)$ and an interval $I\subseteq [|w|]$
which is wpa w.r.t. both $\cpli$ and $\cplii$, we denote by $\omega_{|I}$
the $2$-nested word consisting of $w_{|I}$ and of the two matchings
$\cpli'$ and $\cplii'$ obtained from $\cpli$ and $\cplii$ by restricting them to
$I$, and then shifting them to the interval $[|I|]$. It is routine to verify that if
$\omega$ is a $2$-nested word, then so is
$\omega_{|I}$. 
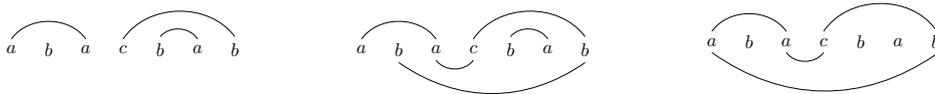
\begin{figure}[h]
\centering
\begin{subfigure}[b]{.3\textwidth}
\scalebox{0.7}{
\begin{tikzpicture}
\tikzset{node distance=0.7cm,}
\node (0) {$a$};
\node[right of = 0] (1) {$b$};
\node[right of = 1] (2) {$a$};
\node[right of = 2] (3) {$c$};
\node[right of = 3] (4) {$b$};
\node[right of = 4] (5) {$a$};
\node[right of = 5] (6) {$b$};
\draw[-] (0.north) edge[bend left =55] (2.north);
\draw[-] (3.north) edge[bend left =55] (6.north);
\draw[-] (4.north) edge[bend left =55] (5.north);

\draw[-, color=white, opacity=0] (1.south) edge[bend right =35] (6.south);
\draw[-,color=white, opacity=0] (2.south) edge[bend right =55] (3.south);

\end{tikzpicture}
}
\end{subfigure}
\hfil
\begin{subfigure}[b]{.3\textwidth}
\scalebox{0.7}{
\begin{tikzpicture}
\tikzset{node distance=0.7cm,}
\node (0) {$a$};
\node[right of = 0] (1) {$b$};
\node[right of = 1] (2) {$a$};
\node[right of = 2] (3) {$c$};
\node[right of = 3] (4) {$b$};
\node[right of = 4] (5) {$a$};
\node[right of = 5] (6) {$b$};
\draw[-] (0.north) edge[bend left =55] (2.north);
\draw[-] (3.north) edge[bend left =55] (6.north);
\draw[-] (4.north) edge[bend left =55] (5.north);

\draw[-] (1.south) edge[bend right =35] (6.south);
\draw[-] (2.south) edge[bend right =55] (3.south);
\end{tikzpicture}
}
\end{subfigure}
\hfil
\begin{subfigure}[b]{.3\textwidth}
\scalebox{0.7}{
\begin{tikzpicture}
\tikzset{node distance=0.7cm,}
\node (0) {$a$};
\node[right of = 0] (1) {$b$};
\node[right of = 1] (2) {$a$};
\node[right of = 2] (3) {$c$};
\node[right of = 3] (4) {$b$};
\node[right of = 4] (5) {$a$};
\node[right of = 5] (6) {$b$};
\draw[-] (0.north) edge[bend left =55] (2.north);
\draw[-] (3.north) edge[bend left =55] (6.north);

\draw[-] (0.south) edge[bend right =35] (6.south);
\draw[-] (2.south) edge[bend right =55] (3.south);
\end{tikzpicture}
}
\end{subfigure}
\caption{A nested word (left) and two $2$-nested words (middle and right).}
\label{fig:2nw}
\end{figure}

\paragraph*{Waves and wave words}
In the sequel, we introduce the restriction of $2$-nested words on which we will focus.
Intuitively, wave structures are graphs obtained from the two matchings consisting of cycles alternating $\cpli$-arches and $\cplii$-arches whose shape evokes waves. 
\begin{definition}
Let $n$ be an integer, and  $(\cpli,\cplii)$ be a
pair of matching relations of length $n$.
A sequence of $4$ integers $1\leq i_1 < i_2 < i_3 < i_4 \le n$
is a \emph{$2$-wave} 
if
the following holds:
\begin{itemize}
\item $\cpli(i_1,i_2)$ and $\cpli(i_3,i_4)$ (top arches),
\item $\cplii(i_2,i_3)$ (bottom arch), and $\cplii(i_1,i_4)$ (support arch)
\end{itemize}
A pair $(\cpli,\cplii)$ is a \emph{2-wave structure} if any arch in $\cpli\cup \cplii$ belongs to a $2$-wave.  
\end{definition}
\begin{remark}
One could allow 2-wave structures to admit 1-waves, 
\emph{i.e.} pairs of indices $(i_1,i_2)$ with  $i_1<i_2$, $\cpli(i_1,i_2)$ and
$\cplii(i_1,i_2)$. All our results would also hold for this generalization.
However, in order to simplify the presentation of the paper, we do not
consider them in this extended abstract.
\end{remark}

\begin{definition}
A $2$-wave word is a 2-nested word $\omega = (w,\cpli, \cplii)$ such that $(\cpli,\cplii)$ is a $2$-wave structure. We denote by $\ww_2(\Sigma)$  
the set of $2$-wave words  over the alphabet $\Sigma$.
\end{definition}

\begin{example}
Examples of waves are given on Figure~\ref{fig:waveintro}.
Let us consider the 2-nested words depicted on Figure~\ref{fig:2nw}.
The one on the middle is not a 2-wave word (the upper arch $(5,6)$ does not belong to a $2$-wave), while the one on the right is.
\end{example}

\paragraph*{Grammar}
In order to proceed with structural induction, we present
an inductive presentation of 2-wave words based on
multiple context-free grammars (MCFG for short~\cite{Seki91}). To this end, we turn to
2-visibly pushdown languages: the alphabet $\Sigma$ is duplicated into
five copies $\Sigma^c_c$, $\Sigma^c_r$, $\Sigma^r_c$, $\Sigma^r_r$ 
and $\Sigma^\textup{int}_\textup{int}$, whose disjoint 
is denoted $\tilde{\Sigma}$. This way, the
two matchings are encoded in the types of the symbols: the upper index is for the first and the lower index for the second matching relation. More formally,
given $\omega \in \ww_2(\Sigma)$, we denote by $\tilde{\omega}\in \tilde{\Sigma}^*$
its visibly pushdown version. 

In MCFG, non-terminals allow to express tuples of words. We present
a grammar with two non-terminals $\mathsf{W}$
and $\mathsf{H}$ which represent respectively words (denoted $w\in \tilde{\Sigma}^*$), and
pairs of words (denoted $(x,y)\in \tilde{\Sigma}^* \times \tilde{\Sigma}^*$). The grammar is
defined by the following rules:
$$
\begin{array}{lll}
\mathsf{W} \ni w & ::= & \epsilon \mid i  \mid w_1 w_2 \mid x w y \\
\mathsf{H} \ni (x,y) & ::= & (\epsilon,\epsilon) \mid (x_1x_2,y_2y_1) \mid
(w_1x w'_1, w_2 y w'_2) \mid (axb,cyd) 
\end{array}
$$
where $i\in \Sigma^\textup{int}_\textup{int}$
and 
$(a,b,c,d)\in \Sigma^c_c  \times \Sigma^r_c \times  \Sigma^c_r \times \Sigma^r_r $.

\begin{restatable}{lemma}{twowaves}
\label{lemma:twowaves}
$L(\mathsf{W}) = \{\tilde{\omega}\in \tilde{\Sigma}^* \mid \omega \in \ww_2(\Sigma)\}$
\end{restatable}

\begin{proof}[Proof sketch]
We give some hints on how to show the right to left implication. 
Let $\omega = (w,M_1,M_2) \in \ww_2(\Sigma)$.
We show, by induction on $n\leq |w|$, the following properties:
\begin{itemize}
\item Let $I \subseteq [|w|]$ such that $|I|=n$ and $I$ is wpa, then
$\tilde{\omega}_{|I} \in L(\mathsf{W})$.
\item Let $I_1,I_2 \subseteq [|w|]$ such that $|I_1|+|I_2|=n$
and $(I_1,I_2)$ is wpa, then $(\tilde{\omega}_{|I_1},\tilde{\omega}_{|I_2}) \in L(\mathsf{H})$.
\end{itemize}
The proof decomposes the 2-wave, by distinguishing cases
according to the structure of the two matchings. One can then
verify that in all cases, one can produce the words using one of the rules of the 
 grammar.
\end{proof}

%% file: automata.tex

Nested word automata have been introduced
in~\cite{DBLP:journals/jacm/AlurM09} as an extension of finite-state 
automata intended to recognize nested words. 
They label arches of the matching relation with so-called \emph{hierarchical}
states:
if the position corresponds to a call (resp. a return), then the automaton
"outputs" (resp. "receives") the hierarchical state used to label the arch,
hence we place it after the input letter (resp. before).
This corresponds to push/pop operations
performed by a (visibly) pushdown automaton. 
The extension of this model to multiple matchings is natural, and
has already been considered in~\cite{DBLP:journals/lmcs/Bollig08,DBLP:journals/tcs/CarotenutoMP16}: 
with two matchings, automata label edges of both
matchings with hierarchical states.

We first introduce some notations. Let us assume that two matching
relations $\cpli$, $\cplii$ of length $n$ are given.
Then, each index $i\in [n]$
can be labelled in 9 different ways depending on it's call, return and internal 
status with regard to matchings $\cpli$ and $\cplii$. 
We say that a position $i$ is a \emph{call-return} if it is a call w.r.t. $\cpli$,
and a return w.r.t. $\cplii$. We extend this convention to other
possible types of positions (\emph{call-call}, \emph{return-call}, \emph{call-internal}\ldots).
Let $I\subseteq[n]$ and $x, y\in\set{c, r, \text{int}}$. Following our graphical representation of 2-nested words, in which
$\cpli$ is depicted above the word, and $\cplii$ is depicted below, 
we denote by $I^x_y$  the subset of $I$ with $x$ status on $\cpli$, and 
$y$ status on $\cplii$. For instance, given a position $i\in I$, we have
$i\in I^c_r$ if $i$ is a call w.r.t. $\cpli$ and a return w.r.t. $\cplii$.
We also introduce the following shortcuts:
for $x\in\set{c, r}$, $I^x = I^x_c \cup I^x_r \cup I^x_{\text{int}}$
and $I_x = I^c_x \cup I^r_x \cup I^{\text{int}}_x$.
For instance, $I^c$ (resp. $I_c$) denotes the set of positions in $I$ which are a call w.r.t.
$\cpli$ (resp. $\cplii$).

\begin{definition}
A \emph{2-nested word automaton} is a tuple $A=(Q, Q_0, Q_f, P, \Sigma, \Delta)$ where :
\begin{itemize}
\item $Q$ is a finite set of states and $Q_0, Q_f\subseteq Q$ are respectively the initial and final states
\item $P$ is a set of hierarchical states
\item $\Delta=(\Delta^x_{y})_{x, y \in\{c, r, \text{int}\}}$ is a set of transitions : 
$\Delta^x_{y}             \subseteq Q \times P^{\mathsf{in}_{x,y}} \times\Sigma\times P^{\mathsf{out}_{x,y}} \times Q$, where $\mathsf{in}_{x,y}$ (resp. $\mathsf{out}_{x,y}$) 
is the number of $r$ (resp. $c$) in $\{x,y\}$. For instance, $\mathsf{in}_{c,r}=\mathsf{in}_{r,c}=1$  and $\mathsf{out}_{c,c}=2$.
\end{itemize}
\end{definition}

\begin{remark}
In the previous definition, elements in $P^{\mathsf{in}_{x,y}}$
correspond to hierarchical states that label closing arches, which can be interpreted 
as popped symbols, while elements in $P^{\mathsf{out}_{x,y}}$
correspond to hierarchical states that label opening arches, which can be interpreted 
as pushed symbols.
Elements of $Q$ are called \emph{linear} states,
as they follow the edges of the linear order, in contrast to hierarchical states.
\end{remark}

\begin{definition}[Run/Language of a \tnwa]
Let $\omega=(a_1\dots a_n, \cpli, \cplii)\in \tnw$ and $A$ be a \tnwa. 
Let  $\ell = (\ell_i)_{i\in\Ientff{0}{n}}$ be a sequence of states,
$h^1 = (h^1_i)_{i\in [n]^c}$ and
$h^2 = (h^2_i)_{i\in[n]_c}$
be two sequences of elements of $P$.
  For all $i\in[n]$, we write $\text{run}^A_i(\omega, \ell,h^1,h^2)$ 
  if one of the following cases holds: 
 (the first four cases are illustrated on Figure~\ref{fig:transition})
\begin{description}
\item {\bf call-call:}  $i\in [n]^c_c $,  and $(\ell_{i-1}, a_i, h^1_i,h^2_i, \ell_i)   \in \Delta^c_{c}$
\item {\bf return-call:} $ i\in [n]^r_c$, and $(\ell_{i-1}, h^1_{\imatchun{i}}, a_i, h^2_i, \ell_i) \in \Delta^r_c$
\item {\bf call-return:} $ i\in [n]^c_r$, and $(\ell_{i-1}, h^2_{\imatchdeux{i}}, a_i, h^1_i, \ell_i) \in \Delta^c_r$
\item {\bf return-return:} $ i\in [n]^r_r$, and $(\ell_{i-1}, h^1_{\imatchun{i}},h^2_{\imatchdeux{i}}, a_i,  \ell_i) \in \Delta^r_r$
\item {\bf call-internal:} $i\in [n]^c_\text{int}$,  and $(\ell_{i-1}, a_i, h^1_i, \ell_i)   \in \Delta^c_{\text{int}}$
\item {\bf internal-call:} $i\in [n]^\text{int}_c$,  and $(\ell_{i-1}, a_i, h^2_i, \ell_i)   \in \Delta^{\text{int}}_{c}$
\item {\bf return-internal:} $i\in [n]^r_\text{int} $,  and $(\ell_{i-1}, h^1_{\imatchun{i}}, a_i, \ell_i)   \in \Delta^r_{\text{int}}$
\item {\bf internal-return:} $i\in [n]^\text{int}_r $,  and $(\ell_{i-1}, h^2_{\imatchdeux{i}}, a_i, \ell_i)   \in \Delta^\text{int}_{r}$
\item {\bf internal-internal:} $i\in [n]^\text{int}_\text{int} $,  and $(\ell_{i-1}, a_i, \ell_i)   \in \Delta^\text{int}_{\text{int}}$
\end{description}
\begin{figure}
\begin{subfigure}[b]{.2\textwidth}
\scalebox{0.7}{
\begin{tikzpicture}
\tikzset{node distance=1cm,}
\node(0a){$a_{i-1}$};
\node (0)[below of = 0a, yshift = +0.7cm] {$\bullet$};
\node (0l)[above of = 0, yshift = - 1.2 cm] {${\ell_{i-1}}$};
\node[right of = 0l] (1l) {$\ell_i$};
\node[right of = 0] (1) {$\bullet$};
\node[right of = 0a] (1a) {$a_i$};
\node[right of = 1] (2) {};
\node[above of = 2] (3) {};
\node[below of = 2] (4) {};
\draw[-] (1a.north) edge[bend left =30] node[below,yshift =+0.6cm]{$h^1_i$} (3.south);
\draw[-] (1l.south) edge[bend right =30] node[below,yshift =-0.05cm]{$h^2_i$}(4.north);
\draw[-] (0.east) edge (1.west);
\node(deltarc)[below of = 1,yshift =-1cm ]{$(\ell_{i-1}, a_i, h^1_i,h^2_i, \ell_i)   \in \Delta^c_{c}$};
%
\end{tikzpicture}
}
\end{subfigure}
\hfil
\begin{subfigure}[b]{.23\textwidth}
\scalebox{0.7}{
\begin{tikzpicture}
\tikzset{node distance=1cm,}
\node(0a){$a_{i-1}$};
\node (0)[below of = 0a, yshift = +0.7cm] {$\bullet$};
\node (0l)[above of = 0, yshift = - 1.2 cm] {${\ell_{i-1}}$};
\node[right of = 0l] (1l) {$\ell_i$};
\node[right of = 0] (1) {$\bullet$};
\node[right of = 0a] (1a) {$a_i$};
\node[right of = 1] (2) {};
\node[above of = 0] (3) {};
\node[below of = 2] (4) {};
\draw[-] (1a.north) edge[bend right =30] node[below,yshift =+0.7cm]{$h^1_{M_1^{-1}(i)}$} (3.south);
\draw[-] (1l.south) edge[bend right =30] node[below,yshift =-0.05cm]{$h^2_i$}(4.north);
\draw[-] (0.east) edge (1.west);
\node(deltacc)[below of = 1,yshift =-1cm]{$(\ell_{i-1}, h^1_{\imatchun{i}}, a_i, h^2_i, \ell_i) \in \Delta^r_c$};
%
\end{tikzpicture}
}
\end{subfigure}
\hfil
\begin{subfigure}[b]{.23\textwidth}
\scalebox{0.7}{
\begin{tikzpicture}
\tikzset{node distance=1cm,}
\node(0a){$a_{i-1}$};
\node (0)[below of = 0a, yshift = +0.7cm] {$\bullet$};
\node (0l)[above of = 0, yshift = - 1.2 cm] {${\ell_{i-1}}$};
\node[right of = 0l] (1l) {$\ell_i$};
\node[right of = 0] (1) {$\bullet$};
\node[right of = 0a] (1a) {$a_i$};
\node[right of = 1] (2) {};
\node[above of = 2] (3) {};
\node[below of = 0] (4) {};
\draw[-] (1a.north) edge[bend left =30] node[below,yshift =+0.7cm]{$h^1_i$} (3.south);
\draw[-] (1l.south) edge[bend left =30] node[below,yshift =-0.05cm]{$h^2_{M_2^{-1}(i)}$}(4.north);
\draw[-] (0.east) edge (1.west);
\node(deltacc)[below of = 1,yshift =-1cm ]{$(\ell_{i-1}, h^2_{\imatchdeux{i}}, a_i, h^1_i, \ell_i) \in \Delta^c_r$};

\end{tikzpicture}
}
\end{subfigure}
\hfil
\begin{subfigure}[b]{.25\textwidth}
\scalebox{0.7}{
\begin{tikzpicture}
\tikzset{node distance=1cm,}
\node(0a){$a_{i-1}$};
\node (0)[below of = 0a, yshift = +0.7cm] {$\bullet$};
\node (0l)[above of = 0, yshift = - 1.2 cm] {${\ell_{i-1}}$};
\node[right of = 0l] (1l) {$\ell_i$};
\node[right of = 0] (1) {$\bullet$};
\node[right of = 0a] (1a) {$a_i$};
\node[right of = 1] (2) {};
\node[above of = 0] (3) {};
\node[below of = 0] (4) {};
\draw[-] (1a.north) edge[bend right =30] node[below,yshift =+0.7cm]{$h^1_{M_1^{-1}(i)}$} (3.south);
\draw[-] (1l.south) edge[bend left =30] node[below,yshift =-0.05cm]{$h^2_{M_2^{-1}(i)}$}(4.north);
\draw[-] (0.east) edge (1.west);
\node(deltacc)[below of = 1,yshift =-1cm ]{ $(\ell_{i-1}, h^1_{\imatchun{i}},h^2_{\imatchdeux{i}}, a_i,  \ell_i) \in \Delta^r_r$};

\end{tikzpicture}
}
\end{subfigure}
\caption{Examples of transition steps.}
\label{fig:transition}
\end{figure}
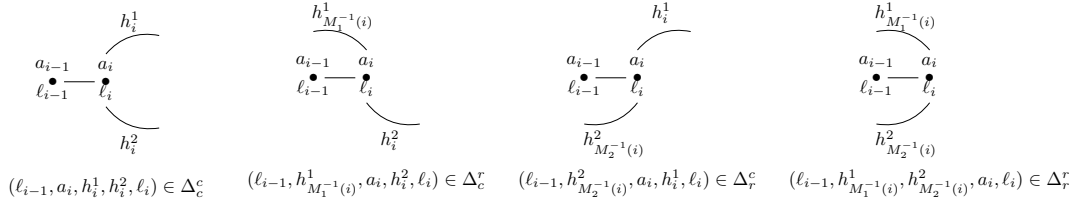

If for all $i\in[n]$,  $\text{run}^A_i(\omega, \ell, h^1, h^2)$ holds, the triple 
$(\ell,h^1,h^2)$ is said to be a \emph{run of $A$ over $\omega$}; it is an \emph{accepting run} if  $\ell_0\in Q_0$ and $\ell_n\in Q_f$.

We will write 
$\exec{q}{\omega}{q'}{A}$ when there exists a triple $(\ell,h^1,h^2)$ 
which is a run of $A$ on $\omega$, and whose first (resp. last) element
of $\ell$ is equal to $q$ (resp. $q'$). If we denote by $n$ the length of $\omega$, and
if $I\subseteq [n]$ is an interval wpa, then we write $\exec{q}{\omega, I}{q'}{A}$
as a shortcut for $\exec{q}{\omega_{|I}}{q'}{A}$.

A 2-nested word is \emph{accepted by $A$} if it admits an accepting run. The set of all 2-nested words accepted by $A$ is denoted $L(A)$, and a language of 2-nested words is called \emph{regular} if it is accepted by a 2-nested word automaton. 
In the sequel, we will also be interested in restricting the language of a \tnwa
to 2-wave words. Hence, we denote by $L_{\ww_2}(A)$ 
 the set of 2-wave words  accepted by
$A$, \emph{i.e.} $L(A)\cap\ww_2(\Sigma)$.
\end{definition}

\begin{example}\label{ex:Aex}
Let $A_{ex}=(Q, Q_0, Q_f, P, \{a,b,c,d\}, \Delta)$, with 
$Q=\{q_a,q_b,q_c,q_d\}$, $Q_0=\{ q_a\}$, $Q_f=\{q_d\}$,
$P=Q$ and $\Delta$ defined as follows (we only give non-empty transition sets): 

\begin{tabular}{l}
$\trans{c}{c}=\{(q_a, a, q_a, q_a ,q_a) \}$\\ 
$\trans{r}{c}=\{ (q_x, q_a, b, q_b, q_b) \mid x\in \{a,b\}\}$\\  
$\trans{c}{r}=\{(q_x, q_b, c, q_c,  q_c) \mid x\in \{b,c\}\}$\\ 
$\trans{r}{r}=\{ (q_x, q_c, q_a, d, q_d) \mid x\in \{c,d\}\}$
\end{tabular}
\end{example}

\begin{wrapfigure}{r}{0.4\textwidth}
 \centering
 \vspace{-1.8cm}
 \input{fig-2NWA}
\end{wrapfigure}
We illustrate the semantics of \tnwa by giving on Figure~\ref{fig:cyclicword} a graphical representation
of a run of $A_{ex}$ on the $\tnw$ 
$\omega_2=(a^{2}b^{2}c^{2}d^{2}, M_1, M_2)$, 
with $M_1, M_2$ depicted on Figure~\ref{fig:cyclicword}.
We let the reader check that the projection
 of $L(A_{ex})$ on $\Sigma^*$ is equal to $\{a^nb^nc^nd^n\mid n \ge 1\}$, 
 hence not context-free.

\begin{definition}[deterministic 2NWA]  A 2NWA  is deterministic iff $Q_0=\set{q_0}$
and
for all $x, y\in\set{c, r,\text{int}}$,  $\trans{x}{y}$ induces a function  $Q \times P^{\mathsf{in}_{x,y}} \times\Sigma \rightarrow  P^{\mathsf{out}_{x,y}} \times Q$.  
\end{definition}

\begin{example}
The automaton $A_{ex}$ considered in Example~\ref{ex:Aex} is deterministic. 
\end{example}

\paragraph*{Normal form}
To ease further constructions, we present a normal form for \tnwa
that requires the hierarchical state of an arch to be equal to the target linear state of its call index.
More formally, we say that a \tnwa is in weakly-hierarchical post form (post form for short)
if  $P=Q$ and for all $x, y\in\set{c, r,\text{int}}$,  $\trans{x}{y} \subseteq \bigcup_{q\in Q} Q \times Q^{\mathsf{in}_{x,y}} \times \Sigma \times  \{q\}^{\mathsf{out}_{x,y}}\times \{q\}$. 
As a consequence, transitions of an automaton in post form can be simplified:    
$\trans{x}{y} \subseteq Q \times Q^{\mathsf{in}_{x,y}} \times \Sigma \times  Q$.

It is worth observing that in a \tnwa in post form, a run is completely characterized by the linear
states. Hence, we can omit the hierarchical states in the formula 
$\text{run}^A_i$, and we can say that a sequence of (linear) states $\ell$ is a run of $A$
on a $2$-nested word $\omega$.
\begin{example}
The automaton $A_{ex}$ considered in Example~\ref{ex:Aex} is in post form.
\end{example}

\begin{restatable}{lemma}{normalform}
\label{lm:post}
Given a \tnwa $A=(Q, Q_0, Q_f, \Sigma, P, \Delta)$, we can build a \tnwa $A' = (Q', q'_0, Q'_f, \Sigma, Q',\Delta')$
which is in weakly-hierarchical post form and such that $L(A)=L(A')$.
\end{restatable}

\paragraph*{Closure properties}
Applying classical automata constructions to \tnwa, one can prove the following
closure properties of regular \tnwl. Detailed proofs can be found in Appendix~\ref{app:closure}.

\begin{restatable}{proposition}{closure}[See also~\cite{DBLP:journals/lmcs/Bollig08}]
\label{proposition:closure}
Regular \tnwl are closed under union, intersection, and direct and reciprocal image by a non erasing alphabetic morphism.  
\end{restatable}

%% file: fig-2NWA.tex


\scalebox{0.9}{
\begin{tikzpicture}
\tikzset{node distance=0.7cm,}
\node (d) {$$};
\node[right of = d] (11) {$a$};
\node[right of = 11] (12) {$a$};
\node[right of = 12] (14) {$b$};
\node[right of = 14] (15) {$b$};
\node[right of = 15] (21) {$c$};
\node[right of = 21] (22) {$c$};
\node[right of = 22] (24) {$d$};
\node[right of = 24] (25) {$d$};

\node[below of = d, yshift =0.4cm] (ds) {$q_a$};
\node[below of = 11, yshift =0.4cm] (11s) {$q_a$};
\node[below of = 12, yshift =0.4cm] (12s) {$q_a$};
\node[below of = 14, yshift =0.4cm] (14s) {$q_b$};
\node[below of = 15, yshift =0.4cm] (15s) {$q_b$};
\node[below of = 21, yshift =0.4cm] (21s) {$q_c$};
\node[below of = 22, yshift =0.4cm] (22s) {$q_c$};
\node[below of = 24, yshift =0.4cm] (24s) {$q_d$};
\node[below of = 25, yshift =0.4cm] (25s) {$q_d$};


\draw[-] (11.north) edge[bend left =40]  node[yshift=0.12cm]{$q_a$}  (15.north);
\draw[-] (12.north) edge[bend left =40] node[yshift=0.12cm]{$q_a$} (14.north);

\draw[-] (21.north) edge[bend left =40] node[yshift=0.12cm]{$q_c$} (25.north);
\draw[-] (22.north) edge[bend left =40] node[yshift=0.12cm]{$q_c$} (24.north);


\draw[-] (11s.south) edge[bend right =40]node[yshift=-0.12cm]{$q_a$} (25s.south);
\draw[-] (12s.south) edge[bend right =40]node[yshift=-0.12cm]{$q_a$} (24s.south);
\draw[-] (14s.south) edge[bend right =40]node[yshift=-0.12cm]{$q_b$} (22s.south);
\draw[-] (15s.south) edge[bend right =40]node[yshift=-0.12cm]{$q_b$} (21s.south);

\end{tikzpicture}}
\caption{A run of $A_{ex}$ over $\omega_2$.}
\label{fig:cyclicword}

%% file: determinization.tex



\begin{theorem}\label{determinisation 2-vagues}
\tnwa are determinizable on the subclass of 2-wave words, \emph{i.e.}, given
a \tnwa $A$, we can build a deterministic \tnwa $A'$ such that
$L_{\ww_2}(A) = L_{\ww_2}(A')$.
\end{theorem}

Thanks to Lemma~\ref{lm:post}, we start from 
a \tnwa in post normal form $A=(Q, Q_0, Q_f, \Delta)$.
This normal form will allow to lighten the presentation, 
as less arguments are needed to write the transitions
and the runs.
We will describe the construction of a deterministic
\tnwa $A'$ (not in post normal form), which accepts 
the same language of 2-wave words.

\paragraph*{Introduction to the construction}
The determinization of finite-state automata, known as the powerset construction, 
registers all the states reachable by a run of $A$, which in the sequel 
are named \emph{current} states.
The determinization procedure of~\cite{DBLP:journals/jacm/AlurM09} for nested word automata 
requires the recording of two states. In addition to the current state, they also store 
the state labelling the call of the arch covering the current position. This state
is named \emph{reference} state. Intuitively, it stores where we come from, and
will be used when closing the arch to reconcile the global run with what happened
below this arch.

In our setting, as we consider a pair of matchings $(\cpli,\cplii)$ instead of
a single one, we need
to store triples of states, composed of two reference states (one for $\cpli$
and one for $\cplii$), and one current state. Hence, states of $A'$ will be sets of such triples
of states.


The very particular shape of $2$-wave words, and in particular the fact that 
the support arch and the top arches do end up at the same return-return position,
ensure that we can gather the information collected along these two paths
to compute, in a deterministic fashion, the set of possible
current states.

However, when trying to address the setting of $2$-wave words, 
we face another difficulty related to reference states.
Indeed, the computations we will perform on each arch of the $2$-wave word are
somehow disconnected. In order to be sure that they can be reconciled, we will enrich
the reference state of the second top arch of the $2$-wave with the
state of the bottom arch and that of the first top arch. This allows us
to check whether the reference associated with the second
top arch is \emph{compatible} with that chosen for the first top arch.

%
%

\paragraph*{Construction}
%
%
We will define the deterministic \tnwa $A'=(Q',\{q'_0\}, Q_f', P', \delta)$
in the following way.
We first introduce the reference states for $\cpli$ and $\cplii$:~\footnote{Observe that
the reference state for the second top arch is a triple of states, as explained before.}
$$\begin{array}{lll}
\overline{\mathcal{R}}^1 &:= Q & \text{the reference for positions at the surface of the first top arch}\\
\overline{\mathcal{R}}^2 &:= Q^3 &\text{the reference for positions at the surface of the second top arch}\\
\underline{\mathcal{R}}  &:= Q &\text{the reference for $\cplii$}
\end{array}$$
We denote by $\overline{\mathcal{R}}$ the union 
$\overline{\mathcal{R}}^1\cup\overline{\mathcal{R}}^2$.
This allows us to define:
$$\begin{array}{lllll}
Q'   &:=2^{\overline{\mathcal{R}}^1\times \underline{\mathcal{R}} \times Q}\cup 2^{\overline{\mathcal{R}}^2\times \underline{\mathcal{R}} \times Q}& \quad &
q'_0 &:=\set{(q_0,q_0,q_0);\ q_0\in Q_0}\\
Q'_f &:=\set{S\in Q';\ S\cap Q\times Q\times Q_f \neq \emptyset}& \quad &
P'   &:= (Q'\times \Sigma) \cup (Q'\times \Sigma)^2
\end{array}$$

\begin{figure}
\input{jolidessin}
\caption{Illustration of the determinization. The (original) run in $A$ is depicted in red, while 
the (new) one in $A'$ is in blue. Elements of states of $A'$ are illustrated by triples 
depicted 
below them.
}\label{fig:construction}
\end{figure}
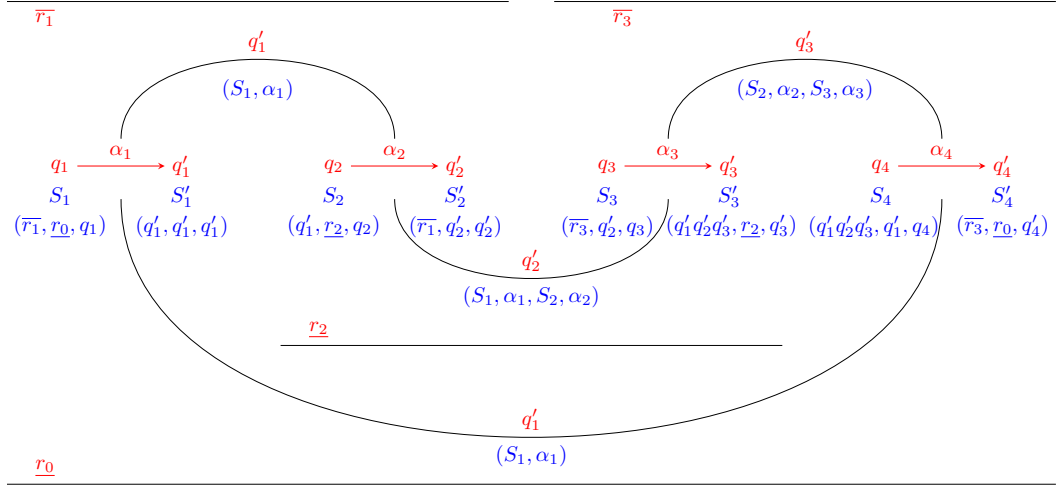
The transition function $\delta$ for $A'$ is defined as follows, by distinguishing cases according to the nature of the symbol. The construction is illustrated on Figure~\ref{fig:construction}. 
In order to ease the writing, the arguments of the formula are not explicitly written.

\begin{itemize}
\item $\delta_\text{int}(S, \alpha) := \set{ (\overline{r}, \underline{r}, q') ;\ \exists q \ (\overline{r}, \underline{r}, q)\in S \land  (q,\alpha,q')\in \Delta_\text{int}}$ 


\item $\delta_c^c(S_1, \alpha_1) := (S'_1, (S_1,\alpha_1), (S_1,\alpha_1))$ where 
\begin{align*}
S_1'   &:= \set{(q_1', q_1', q_1') ;\ \exists \overline{r}_1, \underline{r}_0, q_1 \ \phi_1} \\
\phi_1 &:= (\overline{r}_1, \underline{r}_0, q_1) \in S_1 \land (q_1,\alpha_1, q_1')\in \Delta_c^c.
\end{align*}

\item $\delta^r_c(S_2, (S_1,\alpha_1), \alpha_2) := (S'_2, (S_1,\alpha_1, S_2,\alpha_2))$ where  
 \begin{align*}
 S'_2   &:= \set{(\overline{r}_1, q'_2, q'_2); \ \exists \underline{r}_0, \underline{r}_2, q_1, q'_1,  q_2 \ \phi_2 }\\
 \phi_2 &:= \phi_1 \land (q'_1,\underline{r}_2,q_2) \in S_2 \land (q_2,q'_1,\alpha_2,q'_2)\in \Delta^r_c  
 \end{align*}

\item $\delta^c_r(S_3, (S_1,\alpha_1,S_2,\alpha_2),\alpha_3) := (S'_3, (S_2,\alpha_2,S_3,\alpha_3))$ where~\footnote{$\preceq$ denotes the prefix partial order on strings.}  
 \begin{align*}
 S'_3 &:= \set{(q'_1q'_2q'_3, \underline{r}_2, q'_3); \ \exists \overline{r}_1,  \overline{r}_3, \underline{r}_0,  q_1,q_2, q_3 \ \phi_3 }\\
 \phi_3 &:= \phi_2 \land  \overline{r}_1\preceq \overline{r}_3 \land (\overline{r}_3, q'_2,q_3) \in S_3 \land (q_3,q'_2,\alpha_3,q'_3)\in \Delta^c_r
 \end{align*}

 
\item $\delta^r_r(S_4, (S_2,\alpha_2,S_3,\alpha_3), (S_1,\alpha_1),\alpha_4) :=  S'_4$ where  
 \begin{align*}
 S'_4 &:= \set{(\overline{r}_3, \underline{r}_0, q'_4); \ \exists  \overline{r}_1,\underline{r}_2, q_1, q'_1, q_2, q'_2, q_3, q'_3, q_4  \ \phi_4 }\\
 \phi_4 &:= \phi_3 \land  (q'_1 q'_2 q'_3, q'_1, q_4) \in S_4  \land (q_4, q'_3, q'_1, \alpha_4, q'_4)\in \Delta^r_r  
 \end{align*}

\end{itemize}
\begin{remark}
Formulas $(\phi_j)_{j\in [4]}$ 
are not quantified at all, hence all their variables are free. Their objective is to link parameters extracted from the triplets of states of $A'$.
\end{remark}

\paragraph*{Proof of correctness}
We fix a 2-wave word $\omega=(w, \cpli, \cplii)$, with $w = a_1\dots a_n$.
The following proposition states that on an interval without pending arch, 
runs of $A'$ exactly capture possible runs of $A$, while keeping track
of the reference states, as explained before. 
\begin{restatable}{proposition}{determinisation}
\label{lemme de determinisation}
Let $\Ientff{s}{f} \subseteq [n]$ be wpa, 
and $S, S' \in Q'$ such that $\exec{S}{\omega, \Ientff{s}{f}}{S'}{A'}$, then:
$$\forall  \overline{r} \in \overline{\mathcal{R}},\ \underline{r} \in \underline{\mathcal{R}},\ q'\in Q, \ \left((\overline{r},\underline{r},q')\in S' \iff \exists q \in Q \ (\overline{r},\underline{r},q)\in S \land
 \exec{q}{\omega, \Ientff{s}{f}}{q'}{A}\right)$$
\end{restatable}

Before sketching the proof of this proposition, we fix some notations.
Let $\Ientff{s}{f} \subseteq [n]$ be an interval wpa, 
and $S, S' \in Q'$ such that $\exec{S}{\omega, \Ientff{s}{f}}{S'}{A'}$.
Observe that as $A'$ is deterministic, it has a unique
run on the $2$-wave word $\omega$ restricted to $\Ientff{s}{f}$.
We let $(L_i)_{i\in \Ientff{s-1}{f}}$ denote the (linear) states of this run.
In particular, we have $L_{s-1}=S$ and $L_f=S'$. As $A'$ is deterministic,
hierarchical states are completely determined from linear ones.

In order to prove Proposition \ref{lemme de determinisation}, we separate the direct from the 
indirect case. For the indirect case, it is easy to show by induction that 
any run of $A$ on $\omega$ (restricted to $\Ientff{s}{f}$) will appear in the run of $A'$ on $\omega$ 
(restricted to $\Ientff{s}{f}$). This only
requires to define carefully the corresponding reference states.
\begin{restatable}{lemma}{indirect}
Let $\ell = (\ell_i)_{i\in \Ientff{s-1}{f}}$ be a run of $A$ on $\omega$ (restricted to $\Ientff{s}{f}$) 
such that $\ell_{s-1} = q$, $\ell_{f} = q'$ and $(\overline{r}, \underline{r}, q)\in S=L_{s-1}$. 
Then we can define reference mappings $\overline{R}$ and $\underline{R}$
from $\Ientff{s-1}{f}$ to $\overline{\mathcal{R}}$
and $\underline{\mathcal{R}}$ respectively, such that 
$\overline{R}(f)=\overline{r}$, $\underline{R}(f)=\underline{r}$, 
and for all $i\in \Ientff{s-1}{f}$,
$(\overline{R}(i),\underline{R}(i),\ell_i)\in L_i$.
\end{restatable}



The direct implication will be proven by induction on $\Ientff{s}{f}$.
More precisely, we rely on the grammar allowing to describe all 
$2$-wave words given in Section~\ref{sec:preliminaries}. 
\todo{New sentence here}
We use a slight abuse of notation here,
as the grammar is based on the visibly pushdown presentation, while we work here
with $2$-nested words, but we believe the correspondence is without ambiguity.
The three first cases of the grammar (internal, empty word, and concatenation of
$2$-wave words) are easy. The last case  decomposes 
the $2$-wave word $\omega$ as
$\omega = x \omega' y$, with $(x,y)$ being a pair of "matched" words, \emph{i.e.} given by the non-terminal $\mathsf{H}$. Intuitively, the grammar gives a decomposition of the wpa interval associated with $\omega$ into a wpa interval corresponding to $\omega'$,
and a pair of matched intervals corresponding to $(x,y)$.
Hence, this requires to study pairs of matched intervals, \emph{i.e.} a pair of intervals which is without pending arch.
To this end, we introduce a notation for a "partial" run on a pair of matched intervals:
\begin{definition}
Let $(I_1= \Ientff{i_1}{j_1}, I_2= \Ientff{i_2}{j_2})$ be a pair of matched intervals w.r.t. $\omega$. Let
$q_1, q_2,q'_1,q'_2$ be four states. We let $J = I_1 \cup I_2 \cup \set{i_1-1,i_2-1}$.
We write 
$\exec{q_1, q_2}{\omega, I_1, I_2}{q'_1, q'_2}{A}$
if there exists $(\ell_i)_{i\in J}$ such that $\forall i\in J$, $\text{run}^A_i(\omega, \ell)$ and $\forall k\in [2]$, $q_k=\ell_{i_k-1}$ and $q'_k=\ell_{j_k}$
\end{definition}

We are now ready to state the following lemma:
\begin{restatable}{lemma}{direct}\label{determinisation non wpa} Let $(i_1, i_2, i_3, i_4)$ be a 2-wave such that $i_1\in\Ientff{s}{f}$ and $(\overline{r}_3, \underline{r}_0, q'_4)\in L_{i_4}$. For any $\overline{r}_1, \underline{r}_2, q_1, q'_1, q_2, q'_2, q_3, q'_3, q_4$ satisfying $\psi_4$ we have 
$\exec{q_1,q_3}{\omega,  \Ientff{i_1}{i_2}, \Ientff{i_3}{i_4}}{q'_2,q'_4}{A}$.
\end{restatable}
Lemma \ref{determinisation non wpa} confirms the intuition that a 2-wave is an encapsulation in its sort: any $q_1$, $q'_2$ and $q_3$ yielded by $(\overline{r}_3, \underline{r}_0, q'_4)\in L_{i_4}$ via the construction of $L_{i_4}$ (e.g. via $\psi_4$) define with $q'_4$ a run of $A$ on $\omega$ on the subset of positions given by $\Ientff{i_1}{i_2}$ and $\Ientff{i_3}{i_4}$.


We explain how to conclude the proof of Proposition~\ref{lemme de determinisation}.
Starting from the $2$-wave word $\omega$, we obtained a decomposition
$\omega = x \omega' y$, with $(x,y)$ being a pair of "matched" words produced
by $\mathsf{H}$.
We can show in addition that extremal positions of $x$ and $y$ exactly 
correspond to a $2$-wave, in the sense of the premises
of Lemma~\ref{determinisation non wpa}. Combining it with the induction
hypothesis (direct implication of Proposition~\ref{lemme de determinisation}) applied on $\omega'$,
we can exhibit the expected run in $A$.


\begin{proof}[Proof of Theorem~\ref{determinisation 2-vagues}]
First, observe that by construction, $A'$ is deterministic and complete.
Let $\omega \in \ww_2(\Sigma)$. As $A'$ is deterministic and complete, it has a unique
run on $\omega$ starting from $q'_0$.
Let $(L_i)_{i\in\Ientff{0}{n}}$ be this run, with $L_0=q'_0$.
 We proceed by equivalence:
\begin{align*}
\omega\in L(A') &\iff L_n\in Q'_f \ \iff L_n\cap Q\times Q\times Q_f \neq \emptyset
\  \iff \exists (\overline{r}, \underline{r}, q_f)\in L_n\  q_f\in Q_f\\
           &\iff \exists q_0, q_f\in Q_f\  (\overline{r}, \underline{r}, q_0)\in q'_0 \mbox{ and } \exec{q_0}{\omega}{q_f}{A} \ \iff \omega\in L(A).\qedhere
\end{align*}
\end{proof}

\paragraph*{Closure under complementation}
It is proved in~\cite{DBLP:conf/lics/TorreMP07} that regular \tnwl are not closed under complementation, since they are not determinizable. The  definition of   \tnwl  they use is  sightly different from ours because the two matchings do not share positions, but the same negative result can be shown for our definition. 
 For example, there is no deterministic automaton recognizing the language 
 $\{(a^{n+m} b^{m}a^{n}b^{m}a^{n}, \{(i, i+2(n+m))\mid 1\leq i \leq n+m\} , \{(i, i+n+m)\mid 1\leq i \leq n+m\}) \mid n,m\geq 0 \}$. 
 However, thanks to our determinization result for regular $\ww_2$ languages, we get:  
 \begin{proposition}\label{prop:closureComplement} 
 Regular languages of $2$-wave words are closed under complementation. 
 \end{proposition}

%% file: jolidessin.tex

\newcommand{\red}[1]{\textcolor{red}{#1}}
\scalebox{0.8}{
\begin{tikzpicture}
\tikzset{node distance=1cm}
\node (q1){$\red{q_1}$};
\node[right of = q1, xshift = 1 cm] (q1p) {$\red{q_1'}$};

\node[right of = q1p, xshift = 1.5 cm] (q2) {$\red{q_2}$};
\node[right of = q2, xshift = 1 cm] (q2p) {$\red{q_2'}$};

\node[right of = q2p, xshift = 1.5cm] (q3) {$\red{q_3}$};
\node[right of = q3, xshift = 1 cm] (q3p) {$\red{q_3'}$};

\node[right of = q3p, xshift = 1.5cm] (q4) {$\red{q_4}$};
\node[right of = q4, xshift = 1 cm] (q4p) {$\red{q_4'}$};

\draw[->] (q1) edge[color= red] node[above, color =red](a1){$\alpha_1$}  node[below, yshift = -0.3cm](b1){} (q1p);
\draw[->] (q2) edge[color= red]  node[above, color =red](a2){$\alpha_2$} node[below, yshift = -0.3cm](b2){}  (q2p);
\draw[->] (q3) edge[color= red]  node[above, color =red](a3){$\alpha_3$}  node[below, yshift = -0.3cm](b3){} (q3p);
\draw[->] (q4) edge[color= red]  node[above, color =red](a4){$\alpha_4$}  node[below, yshift = -0.3cm](b4){} (q4p);

\draw[-] (b1) edge[bend right=90] node[above, color =red]{$q'_1$}  node[below, color = blue]{$(S_1,\alpha_1)$}  (b4);
\draw[-] (b2) edge[bend right=90] node[above, color =red]{$q'_2$}  node[below, color = blue]{$(S_1,\alpha_1,S_2,\alpha_2)$} (b3);
\draw[-] (a1) edge[bend left=90] node[above, color =red]{$q'_1$} node[below, color = blue, yshift = -0.2cm]{$(S_1,\alpha_1)$} (a2);
\draw[-] (a3) edge[bend left=90] node[above, color =red]{$q'_3$} node[below, color = blue, yshift = -0.2cm]{$(S_2,\alpha_2,S_3,\alpha_3)$} (a4);

\node[below of = q1, color = blue, yshift =0.5 cm] (S1){${S_1}$};
\node[below of = q1p, color = blue, yshift =0.5 cm] (S1p) {${S_1'}$};

\node[below  of = q2, color = blue, yshift =0.5 cm] (S2) {${S_2}$};
\node[below  of = q2p, color = blue, yshift =0.5 cm] (S2p) {${S_2'}$};

\node[below  of = q3, color = blue, yshift =0.5 cm] (S3) {${S_3}$};
\node[below  of = q3p, color = blue, yshift =0.5 cm] (S3p) {${S_3'}$};

\node[below  of = q4, color = blue, yshift =0.5 cm] (S4) {${S_4}$};
\node[below  of = q4p, color = blue, yshift =0.5 cm] (S4p) {$S_4'$};

\node[below of = q1, color = blue, yshift =0 cm] (s1){$(\overline{r_1}, \underline{r_0},q_1)$};
\node[below of = q1p, color = blue, yshift =0 cm] (s1p) {$(q_1',q_1',q_1')$};

\node[below  of = q2, color = blue, yshift =0 cm] (s2) {$(q_1', \underline{r_2},q_2)$};
\node[below  of = q2p, color = blue, yshift =0 cm] (s2p) {$(\overline{r_1}, q_2',q_2')$};

\node[below  of = q3, color = blue, yshift =0 cm] (s3) {$(\overline{r_3}, q_2',q_3)$};
\node[below  of = q3p, color = blue, yshift =0 cm] (s3p) {$\: (q_1'q_2'q_3', \underline{r_2},q_3')$};

\node[below  of = q4, color = blue, yshift =0 cm] (s4) {$(q_1'q_2'q_3', q_1',q_4)\:\:\:$};
\node[below  of = q4p, color = blue, yshift =0 cm] (s4p) {$(\overline{r_3}, \underline{r_0},q_4')$};

\node[below of = a1, yshift = 3.5cm, xshift = -2cm](1){};
\node[below of = a2, yshift = 3.5cm, xshift = 2cm](2){};
\draw[-] (1) edge  node[below, color = red, xshift = -3.5cm]{$\overline{r_1}$}  (2);

\node[below of = a3, yshift = 3.5cm, xshift = -2cm](3){};
\node[below of = a4, yshift = 3.5cm, xshift = 2cm](4){};
\draw[-] (3) edge  node[below, color = red, xshift = -3cm]{$\overline{r_3}$}  (4);

\node[below of = a1, yshift = - 4.5cm, xshift = -2cm](5){};
\node[below of = a4, yshift = - 4.5cm, xshift = 2cm](6){};
\draw[-] (5) edge  node[above, color = red, xshift = -8cm]{$\underline{r_0}$}  (6);

\node[below of = a2, yshift = - 2.2cm, xshift = -2cm](7){};
\node[below of = a3, yshift = - 2.2cm, xshift = 2cm](8){};
\draw[-] (7) edge  node[above, color = red, xshift = -3.5cm]{$\underline{r_2}$}  (8);
\end{tikzpicture}
}
\vspace{-.5cm}

%% file: applications.tex


We show that languages  of $2$-wave words definable in monadic second order logic (MSO) are exactly regular languages of $2$-wave words. Let us  fix $\mathcal V= \mathcal{V}_1 \cup \mathcal{V}_2$, $\mathcal{V}_1$ is the set of first order variables (whose elements will be written using lower-cases) and $\mathcal{V}_2$ is the set of second order variables (whose elements will be written using upper-cases). 

The monadic second-order logic of 2-nested words ($\text{MSO}(\Sigma, <,M_1,M_2)$) is given by  
$$\phi := Q_a(x)\ |\ x < y\ |\ X(x) \ |\ M_i(x, y)\ |\ \non\phi\ |\ \phi \et \phi\ |\ \ilx x\ \phi\ |\ \ilx X\ \phi$$
where 
$a \in\Sigma$, $x, y\in\mathcal{V}_1$, $X \in\mathcal{V}_2$, $i\in \set{1, 2}$.

The semantics is defined over 2-nested words in a natural way. The first-order
variables are interpreted over positions of the nested word, while set variables are interpreted over sets of positions; $Q_a(x)$ holds if the symbol at the position interpreted for $x$ is a, $x < y$  holds if the position interpreted for $x$ is lesser that  the position interpreted for $y$, and $M_i(x,y)$ holds if the positions interpreted for $x$ and $y$ are related by a nesting edge
of matching $M_i$. 

Let us define some fragments of $\text{MSO}(\Sigma, <,M_1,M_2)$. The set $\text{FO}(\Sigma, <,M_1,M_2)$ of first order (FO) formulas  is the set of all  formulas in  $\text{MSO}(<,M_1,M_2)$ that do not contain any second-order quantifier.  Furthermore, the set $\text{EMSO}(\Sigma, <,M_1,M_2)$ of existential MSO (EMSO) formulas consists  of all formulas of the form $\exists X_1 \dots \exists  X_n \phi$, with $\phi \in \text{FO}(\Sigma,<,M_1,M_2)$. 
 If $\mathcal{L}$ is a logic (MSO, EMSO or FO),  and  $\phi$ is a closed formula in $\mathcal{L}(\Sigma, <,M_1,M_2)$, the \emph{language defined by} $\phi$, denoted $L(\phi)$ is the set of all 2-nested words that are a model for $\phi$, and $L_{\ww_2}(\phi)$ is the set of all $2$-wave words
  that are a model for $\phi$, that is,  $L_{\ww_2}(\phi)=L(\phi)\cap \ww_2(\Sigma).$

In~\cite{DBLP:journals/lmcs/Bollig08}, Bollig shows the equivalence between automata and EMSO for the whole
class of $2$-nested words. However, he shows that quantifier alternation yields an infinite hierarchy.
In our setting, this hierarchy collapses, and we obtain the equivalence between MSO and EMSO.
The following characterization extends that given in \cite{DBLP:journals/jacm/AlurM09} for regular nested word languages. Its proof follows classical lines, and relies on Propositions~\ref{proposition:closure} and~\ref{prop:closureComplement}.\todo{j'ai viré le sketch de preuve pour gagner de la place}

\begin{theorem} \label{MSO-waves} A language of $2$-wave words is definable in EMSO iff it is definable in MSO iff it is regular. 
 \end{theorem}

%% file: decision2.tex


The analysis we did so far of \tnwa over $2$-wave words
allows to establish the following decidability results:

\begin{theorem}
Let $A$ and $B$ be two \tnwa over $\Sigma$. The following holds:
\begin{itemize}
\item Determining whether $L_{\ww_2}(A)=\emptyset$ can be decided in polynomial time.
\item Determining whether $L_{\ww_2}(A)=\ww_2(\Sigma)$ can be decided in exponential time.
\item Determining whether $L_{\ww_2}(A) \subseteq L_{\ww_2}(B)$ 
(resp. $L_{\ww_2}(A) = L_{\ww_2}(B)$) can both 
be decided in exponential time.
\end{itemize}
\end{theorem}

\begin{proof}[Proof sketch]
The first result follows from 
the grammar
presented in Section~\ref{sec:preliminaries} that produces all $2$-wave words.
More precisely, we maintain a set $W$ of pairs of states (for non-terminal \textsf{W})
and a set $H$ of quadruplets of states (for non-terminal \textsf{H}). Intuitively, 
$(p,q)\in W$ means that there exists a $2$-wave word $\omega$
and a run $\exec{p}{\omega}{q}{A}$. We initialize them with identity relations and saturate them
using the rules of the grammar, and corresponding transitions of $A$.
The other results are direct consequences from closure under complement
using the determinization procedure, with the observation
that the latter has exponential complexity.
\end{proof}

\paragraph*{On the treewidth of $2$-wave words}

Any $\tnw$ $\omega = (w,M_1,M_2)$ of length $n$ can be seen as a graph whose set of vertices is $[n]$ and set of edges is $M_1\cup M_2 \cup \{(i,i+1)\}_{i\in [n-1]}$, then  an alternative approach to decidability is by using Courcelle's Theorem
on graphs of bounded treewidth. Indeed, using the grammar presented in
Section~\ref{sec:preliminaries} for $2$-wave words 
(Lemma~\ref{lemma:twowaves}), we show:
\begin{restatable}{lemma}{lemmatw}
For all $\omega \in \ww_2(\Sigma)$, 
 $\omega$ has treewidth at most $11$.
\label{lm:tw}
\end{restatable}

In addition, it is easy to verify that the class of $2$-wave words can be expressed in MSO.
As a consequence, we obtain using~\cite{Courcelle97,Seese91}:
\begin{proposition}
Given a formula $\phi\in\text{MSO}(<,M_1,M_2)$, checking whether there exists 
$\omega \in \ww_2(\Sigma)$ that satisfies $\phi$ is decidable.
\end{proposition}

%% file: discussion.tex
Let $\mathcal{L}$ be a logic (FO, EMO, or MSO),  we  denote by   $\exists \text{Match} \mathcal{L}(\Sigma,<, M)$ the class of all word languages $L=\{ u\in\Sigma^* \mid  \exists M  (u,M) \in L_M\}$ such that $L_M$ is definable in  $\mathcal{L}(\Sigma,<,M)$.
Similarly, we   denote by   $\exists \ww_2 \mathcal{L}(\Sigma,<, M_1,M_2)$ the class of languages  $L=\{ u\in\Sigma^* \mid  \exists  (u,M_1,M_2) \in L_{\ww_2}(\phi)\}$ for $\phi \in  \mathcal{L}(\Sigma,<, M_1,M_2)$.

It is proved in   \cite{DBLP:conf/csl/LautemannST94} that  $\text{CFL}=\exists \text{Match FO}(\Sigma,<,M)=\exists \text{Match  MSO}(\Sigma,<,M)$. A key of the proof is the fact that CF grammars can be put in Greibach double normal form. Then each production has the form $X\rightarrow a u b$, and the derivation tree of a word can be encoded by matching letters delimiting each production. 

Consider now indexed grammars that generate Indexed Languages (\il). They are CF grammars where each non-terminal symbol carries a pushdown stack (often denoted $X^\omega$ where $X$ is the non terminal and $\omega$ the stack). Grammars contain \emph{push productions} of the form $X\rightarrow Y^{p}$ saying that any $X^{\omega}$ can be rewritten $Y^{p\omega}$;  \emph{pop productions} of the form $X^p\rightarrow Y$ saying that any $X^{p\omega}$ can be rewritten $Y^{\omega}$ and  \emph{copy productions} of the form $X \rightarrow YZ$  saying that any $X^{\omega}$ can be rewritten $Y^{\omega}Z^{\omega}$. 

The derivation tree of a word can be encoded  using two matching relations: one  delimiting  productions (corresponding to $M_2$), and one for the stack moves   (corresponding to $M_1$). The nested structure thus obtained is  a wave structure,  where each wave follows a pushed symbol amongst the different copies.  In particular, a 2-wave corresponds to a symbol which has been popped but  has never been copied. Indexed grammars  which never process copies are called \emph{linear indexed grammars}  and generate \emph{linear indexed languages} (\lil).  In  such grammars, copy productions have the form $X \rightarrow X^\bullet Y$  or $X \rightarrow XY^\bullet$ where $\bullet$ marks the symbol on which the stack is transmitted.  

The next Proposition establishes a formal connection between
regular languages of $2$-wave words and linear indexed languages: 
\begin{proposition} \label{prop:wavesToLIL}
 Languages in   $\exists \ww_2 \text{MSO}(\Sigma, <, M_1,M_2)$  are linear indexed languages. 
\end{proposition} 
However, we don't know if the equality holds, mainly because we don't know if linear indexed languages can be put in a  Greibach double normal like form.  

We can also ask whether the logical characterization of CFL extends to indexed languages\todo{à partir d'ici ajout: à relire}. The appropriate structure seems to be $\tnw$ whose edges form \emph{waves} of unbounded length. 
Intuitively, a $k$-wave generalizes a $2$-wave as follows: it consists of $2k$ indices in increasing order,
with $k$ top arches, $k-1$ bottom arches, and an additional support arch. An example of a $4$-wave is depicted
on Figure~\ref{fig:waveintro}. This yields the notion of $k$-wave word ($\ww_k$), and that
of wave word ($\ww$), which is simply the union of $\ww_k$ over $k$. As an example, the $2$-nested word
depicted on Figure~\ref{fig:cyclicw} is in $\ww_4$.
So, the question is:  does  $\il=\exists \ww \text{FO}(\Sigma,<,M_1,M_2)=\exists \ww \text{MSO}(\Sigma,<,M_1,M_2)$  hold? 
 Proposition \ref{EMOonINW}  provides an element of response,  and we think that $\il \neq \exists \ww \text{FO}(\Sigma,<,M_1,M_2)$. 

\begin{figure}[h]
\scalebox{0.8}{
\begin{tikzpicture}
\tikzset{node distance=0.7cm,}
\node (d) {$\#$};
\node[right of = d] (1) {${a_1}$};
\node[right of = 1] (2) {$a_2$};
\node[right of = 2] (3) {$a_3$};
\node[right of = 3] (4) {$a_4$};
\node[right of = 4] (d1) {$\#$};
\node[right of = d1] (11) {$a_4$};
\node[right of = 11] (12) {${a_1}$};
\node[right of = 12] (13) {$a_2$};
\node[right of = 13] (14) {$a_3$};
\node[right of = 14] (d2) {$\#$};
\node[right of = d2] (21) {$a_3$};
\node[right of = 21] (22) {$a_4$};
\node[right of = 22] (23) {${a_1}$};
\node[right of = 23] (24) {$a_2$};
\node[right of = 24] (d3) {$\#$};
\node[right of = d3] (31) {$a_2$};
\node[right of = 31] (32) {$a_3$};
\node[right of = 32] (33) {$a_4$};
\node[right of = 33] (34) {$a_1$};
\node[right of = 34] (d4) {$\#$};
\draw[-] (1.north) edge[bend left =30] (4.north);
\draw[-] (2.north) edge[bend left =30] (3.north);

\draw[-] (11.north) edge[bend left =30] (14.north);
\draw[-] (12.north) edge[bend left =30] (13.north);

\draw[-] (21.north) edge[bend left =30] (24.north);
\draw[-] (22.north) edge[bend left =30] (23.north);

\draw[-] (31.north) edge[bend left =30] (34.north);
\draw[-] (32.north) edge[bend left =30] (33.north);


\draw[-] (1.south) edge[bend right =10] (34.south);
\draw[-] (2.south) edge[bend right =10] (33.south);
\draw[-] (3.south) edge[bend right =15] (12.south);
\draw[-] (4.south) edge[bend right =15] (11.south);
\draw[-] (13.south) edge[bend right =30] (22.south);
\draw[-] (14.south) edge[bend right =30] (21.south);
\draw[-] (23.south) edge[bend right =15] (32.south);
\draw[-] (24.south) edge[bend right =15] (31.south);

\end{tikzpicture}}
\caption{Example of a word in $L$ and its associated wave structure.}
\label{fig:cyclicw}
\end{figure}
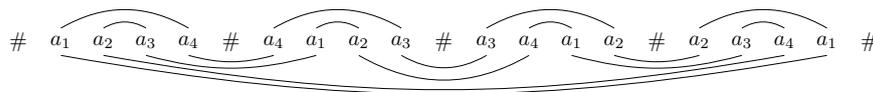

\begin{restatable}{proposition}{EMOonINW}
\label{EMOonINW}
 There exists a language which is not an indexed language but being $\exists \ww \text{EMSO}(\Sigma,<,M_1,M_2)$-definable. 
\end{restatable}
\begin{proof}[Proof sketch]
Consider  $\Sigma=A \cup \set{\#}$ for any alphabet $A$,  
and the set $L$ of all words of the form 
$\#u_1 \# u_2\# \ldots u_{n}\#$, for $n\geq 1$, 
such that for all $i\in [1,n-1]$,   $u_i \in A^n$,  and if 
$u_i = a_1\ldots a_n$, then $u_{i+1}= a_n a_1 \ldots a_{n-1}$. 
Using the Shrinking Lemma given in \cite{DBLP:journals/tcs/Gilman96} for indexed languages, it can easily be proved that $L$ is not an indexed language. In addition, one can write an EMSO-formula defining 2-wave words $(\#u_1 \# u_2\# \ldots u_{n}\#, M_1,M_2)$  such that $(M_1,M_2)$ forms $n/2$ embedded $n$-waves
(see an example on Figure \ref{fig:cyclicw}). 
\end{proof}

%% file: app-grammar.tex


\section{Proof of Lemma~\ref{lemma:twowaves}}

\twowaves*

\begin{proof}
The right to left implication is easy.

We show the left to right implication. Let $\omega=(w,M_1,M_2)$ be a 2-wave.
We show, by induction on $n\leq |u|$, the  following property:
\begin{enumerate}
\item let $I \subseteq [|w|]$ such that $|I|=n$ and $I$ is wpa, then
$\tilde{\omega}_I$ belongs to $L(\mathsf{W})$,
\item let $(I_1,I_2) \subseteq [|w|]$ such that $|I_1|+|I_2|=n$,  
and $(I_1,I_2)$ is wpa, then $( \tilde{\omega}_{|I_1}, \tilde{\omega}_{|I_2})$ belongs to $L(\mathsf{H})$.
\end{enumerate}

For $n=0$ the property follows from the cases $\epsilon$
and $(\epsilon,\epsilon)$ respectively.

Assume the property holds for $n$, and let us show it holds for $n+1$.

We start with the first statement.
We consider a wpa interval $I$ of size $n+1$ and study the letter at the first position (denoted $i$):
\begin{itemize}
\item if it is an internal position, the result easily follows.
\item otherwise, it must be a call-call. Then we consider the matching return-return,
say at position $j$.
Again several cases need to be considered:
\begin{itemize}
\item if $j$ is strictly before the last position of $I$, then the subinterval of $I$
with positions $>j$ is wpa. Hence, induction property can be applied.
\item otherwise, $j$ is at the last position. We can also identify the matching return call (position $k$) and call return (position $\ell$). We define $(I_1,I_2)$ corresponding to the four positions 
$i,k$ and $\ell,j$ respectively, which is wpa.
Again, there are two cases:
\begin{itemize}
\item If the two positions are not consecutive ($k < \ell -1$), then $(I_1,I_2)$ is strictly smaller, so the other property can be applied.
In addition, the interval $[k+1,\ell -1]$ is also wpa. Again, induction property can be applied.
We conclude using the last case of grammar of $\mathsf{W}$.
\item If the two positions are consecutive. We define $(I'_1,I'_2)$ with $I'_1$ starting at $i+1$, ending at $k-1$,
and $I'_2$ starting at $\ell+1$, ending at $j-1$. It is wpa, strictly smaller, induction case 2 can be applied.
Last case of $\mathsf{H}$ can be applied to deduce that $( \tilde{\omega}_{|I_1}, \tilde{\omega}_{|I_2})$ belongs to $L(\mathsf{H})$. We easily conclude 
for $I$ using that $\epsilon$ belongs to $L(\mathsf{W})$, and the last case of $\mathsf{W}$
\end{itemize}
\end{itemize}
\end{itemize}

We turn to the second property.
We consider a pair $(I,J)$, wpa, of length $n+1$. 
If $I=\emptyset$, then $J$ is wpa and $(\tilde{\omega}_{|I},  \tilde{\omega}_{|J})$ is obtained by applying the third case of grammar of $\mathsf{H}$ with $(x,y)=(\epsilon,\epsilon)$, $w_1=w_2=w_3=\epsilon$ and $w_4=  \tilde{\omega}_{|J}$. 

Otherwise, we study the letter at first position of $I$, that we denote by $i$.
If it's an internal, again it's easy using he third case of grammar of $\mathsf{H}$. So we suppose it's a call-call.
Let us denote by $j$ the matching return-return.
We consider two cases:
\begin{itemize}
\item If $j$ belongs to $I$, then we can split $I$ into $I'$ and $I''$, obtained by restricting $I$ to positions less than, or equal to (resp. strictly greater than) $j$. Then, 
$I'$ is wpa, and $(I'',J)$ is too, and of size at most $n$. 
We easily conclude, using induction hypothesis on both, and combining them using third case
of $\mathsf{H}$.
\item Otherwise, $j$ belongs to $J$. We distinguish several cases:
\begin{itemize}
\item if $j$ is not the last position of $J$, then we can split $J$ as we split $I$ before,
and conclude similarly.
\item otherwise, $j$ is at the last position of $J$, and we identify the two other matching positions,
namely the return call $k$, and the call return $\ell$. We distinguish two cases:
\begin{itemize}
\item if $k$ is at the last position of $I$, and $\ell$ is at the first position of $J$,
then we can apply last case of $\mathsf{H}$ to conclude.
\item otherwise, we will be able to split $I$ (resp. $J$)
into $I'$, $I''$ (resp. $J''$, $J'$) such $(I',J')$ and $(I'',J'')$ are both wpa, and non-empty, hence of size
at most n. Using induction hypothesis, we conclude using second case 
of $\mathsf{H}$.
\end{itemize}
\end{itemize}
\end{itemize}

\end{proof}

%% file: app-closure.tex


\section{Proofs of Section~\ref{sec:automata}}
\label{app:closure}

\normalform*
\begin{proof}[Sketch] 
We define $Q'= Q(\{\epsilon\} \cup P \cup PP)$,  $Q'_0=Q_0$, $Q_f=Q_f (\{\epsilon\} \cup P \cup PP)$. 
We construct $\Delta'$ from $\Delta$ such that there is a one to one correspondence between runs of $A$ and  runs of $A'$. More precisely, for all $\omega=(w,M_1,M_2)$,  $(\ell,h^1,h^2)$ is a run of $A$ over $\omega$ iff  $\forall 1\leq i \leq |w|$,  the sequence $\ell'_{i}=\ell_i h^1_{i}h^2_{i}$ is  run of $A'$ over $\omega$ (if $h^1_{i}$ or $h^2_{i} $ is undefined, it is replaced by $\epsilon$). 
Formally,
 for all production $(q_1, p_1, \ldots, p_{\mathsf{in}_{x,y}} a, p'_1, \ldots, p'_{\mathsf{out}_{x,y}}, q_2)\in \trans{x}{y}$, 
$$\{ (q_1', qp_1 \cdots p_{\mathsf{in}_{x,y}}, a, q_2p'_1 \cdots p'_{\mathsf{out}_{x,y}})\mid q_1'\in q_1(\{\epsilon\} \cup P \cup PP), q \in Q \}\in  \Delta'{}^x_y.$$
\end{proof}

\closure*
\begin{proof}
Suppose $L'$ and $L''$ to be recognized respectively by  \tnwa  $A'=(Q', Q_0', Q_f', \Sigma, {\Delta'})$  and  $A''=(Q'', Q_0'', Q_f'', \Sigma, {\Delta''})$ in post form. 

We define the \emph{product} of $A'$ and $A''$ to be  the 2NWA $A =(Q' \times Q'', Q_0'\times Q_0'', Q_f' \times Q_f'', \Sigma, \Delta)$ where  for all $x,y\in \{c,r ,\text{int}\}$,  $\trans{x}{y}$ is the set of transitions 
$$ ((q'_1q''_1), (p_1'p_1''), \ldots , (p'_{\mathsf{in}_{x,y}}p''_{\mathsf{in}_{x,y}}), a ,  (r_1'r_1''), \ldots , (r'_{\mathsf{in}_{x,y}}r''_{\mathsf{in}_{x,y}}),  (q_2'q_2''))$$ such that 
\begin{itemize}
\item $ (q'_1, p_1', \ldots , p'_{\mathsf{in}_{x,y}}, a , r_1', \ldots , r'_{\mathsf{in}_{x,y}},  q_2') \in \transd{\Delta'}{x}{y}$, and 
\item $ (q''_1, p_1'', \ldots , p''_{\mathsf{in}_{x,y}}, a , r_1'', \ldots , r''_{\mathsf{in}_{x,y}},  q_2'') \in \transd{\Delta''}{x}{y}
 .$
\end{itemize}
   Given a 2-nested word $w$ of length $n$, a $Q'$-sequence $(\ell'_i)_{i\in\Ientff{0}{n}}$, a $Q''$-sequence $(\ell''_i)_{i\in\Ientff{0}{n}}$: 
$\ell'$ is a run of $A'$ over $w$ and $\ell''$ is a run of $A''$ over $w$ iff $(\ell'_i,\ell''_i)_{i\in\Ientff{0}{n}}$
is a run of $A$ over $w$. \\


For the union of languages, we suppose that $Q' \cap Q'' \neq \emptyset$ and  define the \emph{sum} of $A'$ and $A''$ to be  the 2NWA $A =(Q'\cup Q'', Q_0'\cup Q_0'', Q_f'\cup Q_f'', \Sigma, \Delta' \cup \Delta'')$.  
 Given a 2-nested word $w$ of length $n$, a $Q'$-sequence $(\ell'_i)_{i\in\Ientff{0}{n}}$, a $Q''$-sequence $(\ell''_i)_{i\in\Ientff{0}{n}}$: 
$\ell'$ is a run of $A'$ over $w$ or  $\ell''$ is a run of $A''$ over $w$ iff $(\ell'_i,\ell''_i)_{i\in\Ientff{0}{n}}$
is a run of $A$ over $w$. \\

For the direct image by non erasing alphabetic morphism $h$, it suffices to replace in transitions each letter $a$ by $h(a)$. For the inverse image, we proceed similarly  but by duplicating transitions: one for each letter $b$ such that $h(b)=a$. 
\end{proof}

%% file: app-deter.tex
\section{Proof of Determinization}\label{sec:proof}

\subsection{Notations}
We start with additional definitions.

\begin{definition}
Let $\cpl$ be a matching relation, $(i,j)\in M$ and $k\in [n]$:
\begin{itemize} 
\item  $k$ is \emph{on the surface} of arch $(i, j)$ (or $(i, j)$ \emph{covers} $k$) iff
$$k\in \Ientfo{i}{j} \ \text{and} \ \qqs i', j'\ (\match{i',j'}\et i<i'\leq k)\imp j'\leq k$$ 
\item if $\match{k,l}$,  arch $(k, l)$ is \emph{on the surface} of  $(i, j)$ iff $l$ is on the surface of $(i, j)$
\end{itemize}
\end{definition}


\begin{definition}
The \emph{surface function} $s_\cpl$ 
yields from any position $k$ (including $0$) the call position of the arch that covers $k$, or $0$ if $k$ is not covered by any arch. In particular, $s_\cpl(0)=0.$
\end{definition}
\begin{definition}
The \emph{surface call-call} $s^c_c$ 
yields from any position $k$ (including $0$) the call-call position of the arch that covers $k$, or $0$ if $k$ is not covered by any arch. In particular, $s^c_c(0)=0.$
\end{definition}
\begin{remark} the reader can check the following properties: 
\begin{itemize}
\item if $k>0$ is internal then  $s_\cpl(k)=s_\cpl(k-1)$
\item if $k>0$ is a call then  $s_\cpl(k)=k$
\item if $k>0$ is a return and $M(i,k)$ then   $s_\cpl(k)=s_\cpl(i-1)$
\item if $i:=s_{\cpli}(k)>0$ then $i$ is in $[n]^c_c$ and $s^c_c(k)=i$ or $i$ is in $[n]^c_r$ and $s^c_c(k)=(\cplii\circ\cpli)^{-1}(i)$
\end{itemize}
\end{remark}

\subsection{Main proof}

\determinisation*
\begin{proof}
Let $\Ientff{s}{f} \subseteq [n]$ be an interval wpa, 
and $S, S' \in Q'$ such that $\exec{S}{\omega, \Ientff{s}{f}}{S'}{A'}$.
Observe that as $A'$ is deterministic, it has a unique
run on the $2$-wave word $\omega_{|\Ientff{s}{f}}$.
We let $(L_i)_{i\in \Ientff{s-1}{f}}$ denote the (linear) states of this run.
In particular, we have $L_{s-1}=S$ and $L_f=S'$. As $A'$ is deterministic,
hierarchical states are completely determined from linear ones.

We start with the proof of the reverse implication:
\indirect*
\begin{proof}

We define: 
$$\begin{array}{rcl}
     \overline{R}: \Ientff{s-1}{f} & \to     & \overline{\mathcal{R}}^1 \cup \overline{\mathcal{R}}^2 \\
                                 i & \mapsto & \left\{\begin{array}{ll}
                                                        \overline{r} & \text{if} \ s_{\cpli}(i) < s \\ 
                                                        \ell_{s_{\cpli}(i)} &\text{if}\ s_{\cpli}(i)\in \Ientff{s}{f}^c_c\\
                                                        \ell_{s^c_c(i)} \ell_{m^{-1}_{\cplii}(s_{\cpli}(i))} \ell_{s_{\cpli}(i)} & \text{if}\ s_{\cpli}(i)\in \Ientff{s}{f}^c_r\\
                                                      \end{array}\right.\\
{}                                 & {}       & {}\\
     \underline{R}: \Ientff{s-1}{f} & \to     & \underline{\mathcal{R}} \\
                                  i & \mapsto & \left\{\begin{array}{ll}
                                                         \underline{r} & \text{if} \ s_{\cplii}(i)< s \\ 
                                                         \ell_{s_{\cplii}(i)} &\text{if}\ s_{\cplii}(i)\in\Ientff{s}{f}^2_c\\
                                                       \end{array}\right.
  \end{array}$$

First, using these definitions, and the fact that $\Ientff{s}{f}$ is wpa, it is easy to
verify that $\overline{R}(f) = \overline{r}$
and $\underline{R}(f) = \underline{r}$.

Second, we observe that if $(i_1, i_2, i_3, i_4)$ is a 2-wave within $\Ientff{s}{f}$ then $\overline{R}(i_2)\preceq\overline{R}(i_3-1)$ :
\begin{itemize}
\item if $s_{\cpli}(i_2)=s_{\cpli}(i_3-1)$ e.g. $i_2$ and $i_3-1$ are covered by the same $\cpli$-arch then $\overline{R}(i_2)=\overline{R}(i_3-1)$
\item if $s_{\cpli}(i_2)<s_{\cpli}(i_3-1)$ then $s_{\cpli}(i_2) = s^c_c(i_3-1)$ e.g. $i_2$ and $i_3-1$ are covered by the same wave (but not the same arch) then:
$$\overline{R}(i_2) = \ell_{s_{\cpli}(i_2)} = \ell_{s^c_c(i_3-1)}\preceq\ell_{s^c_c(i_3-1)} \ell_{m^{-1}_{\cplii}(s_{\cpli}(i_3-1))} \ell_{s_{\cpli}(i_3-1)}=\overline{R}(i_3-1)$$
\end{itemize}

We prove by induction over $k\in  \Ientff{s-1}{f}$ that $(H_{k})$:   $\forall  s-1 \leq i \leq k$, 
$(\overline{R}(i),\underline{R}(i), \ell_i)\in L_i$. 
For $k=s-1$,  the property follows from the hypothesis $(\overline{r},\underline{r},q)\in L_{s-1}$.

Suppose $k\geq s$,  and $( \overline{R}(i), \underline{R}(i), \ell_{i})\in L_{i}$, for all $s-1\leq i\leq k-1$. 

We consider the following cases: 

\begin{itemize}
\item if $k$ is an internal position then from construction of $\delta_{\text{int}}$ and  $(H_{k-1})$:  $L_k$ contains 
$(\overline{R}(k-1),\underline{R}(k-1),\ell_k)$. 
 In addition, $\overline{R}(k)= \overline{R}(k-1)$ and $\underline{R}(k)= \underline{R}(k-1)$. 

\item if $k\in \Ientff{s}{f}^c_c$, then from the construction of $\delta^c_c$ and $(H_{k-1})$,  $(\ell_k, \ell_k, \ell_k)\in L_k$. 
 In addition, $s_{\cpli}(k)=s_{\cplii}(k)=k$ hence $\overline{R}(k) = \underline{R}(k)=\ell_k$. 

 \item if $k\in \Ientff{s}{f}^r_c$, then $k=i_2$ in a wave $(i_1,i_2,i_3,i_4)$ included in
 $\Ientff{s}{f}$ as it is wpa.  Since $\ell$ is a run and 
$(H_{k-1})$ holds, $(\ell_{k-1}, \ell_{i_1}, a_k,\ell_k)\in \trans{r}{c}$,  and 
\begin{equation}
\label{eq1}
(\ell_{i_1-1},a_{i_1},\ell_{i}) \in \trans{c}{c}, \ 
(\overline{R}(i_{1}-1),\underline{R}(i_1-1),\ell_{i_1-1}) \in L_{i_1-1}, \ (\ell_{i_1},\ell_{i_1},\ell_{i_1})\in L_{i_1} 
 \end{equation}
 
 From construction of $\delta^r_c$, $L_k$ contains 
$(\overline{R}(i_1-1), \ell_k, \ell_k)$. 
 In addition, $s_{\cpli}(k)=s_{\cpli}(i_1-1)$ hence $\overline{R}(k)=\overline{R}(i_1-1)$ and $s_{\cplii}(k)=k$. It follows that $\underline{R}(k)=\ell_{k}$.
 
 
\item if $k \in \Ientff{s}{f}^c_r$, then $k=i_3$ in a wave $(i_1,i_2,i_3,i_4)$  included in
 $\Ientff{s}{f}$ as it is wpa. Since $\ell$ is a run and 
$(H_{k-1})$ holds, $(\ell_{k-1}, \ell_{i_2}, a_k,\ell_k)\in  \trans{c}{r}$, (\ref{eq1}), and 
\begin{equation}
\label{eq2}
(\ell_{i_2-1}, \ell_{i_1} ,a_{i_2},\ell_{i_2}) \in  \trans{r}{c}, \ 
( \overline{R}(i_2), \underline{R}(i_2-1), \ell_{i_2-1})\in L_{i_2-1}, \ ( \overline{R}(i_2), \underline{R}(i_2),\ell_{i_2})\in L_{i_2} 
 \end{equation}

From construction of $ \delta^c_r$, 
$(\ell_k, \underline{R}(i_2-1), \ell_k) \in L_k$. 
 In addition, $s_{\cpli}(i_3)=i_3$ hence $\overline{R}(i_3)=\ell_{i_1}\ell_{i_2}\ell_{i_3}$ and $s_{\cplii}(i_3)=s_{\cplii}(i_2-1)$. It follows that $\underline{R}(i_3)=\underline{R}(i_2-1)$.


\item if $k  \in \Ientff{s}{f}^r_r$, then $k=i_4$ in a wave $(i_1,i_2,i_3,i_4)$  included in
 $\Ientff{s}{f}$ as it is wpa. Since $\ell$ is a run and 
$(H_{k-1})$ holds, $(\ell_{k-1}, \ell_{i_3},\ell{i_1}n a_k,\ell_k)\in \trans{r}{r}$,  $(\ell_k, \ell_k, \ell_{k-1})\in L_{k-1}$, (\ref{eq1}), (\ref{eq2}) and 
\begin{equation*}
\label{eq2}
(\ell_{i_3-1}, \ell_{i_2}, a_{i_3},\ell_{i_3}) \in \trans{r}{c}, \ 
(\overline{R}(i_1-1), \ell_k, \ell_{i_3-1}) \in L_{i_3-1}, \ ( \ell_k, \overline{R}(i_2-1), \ell_k)\in L_{i_3} 
 \end{equation*}

From construction of $\delta^r_r$, 
$(\overline{R}(i_1-1), \underline{R}(i_1-1), \ell_k) \in L_{k}$. 
 In addition, $s_{\cpli}(i_4)=s_{\cpli}(i_3-1)$ hence $\overline{R}(i_4)=\overline{R}(i_3-1)$ and $s_{\cplii}(i_4)=s_{\cplii}(i_1-1)$ hence $\underline{R}(i_4)=\underline{R}(i_1-1)$.

\end{itemize}
\end{proof}

We turn to the direct implication, and consider
$\overline{r} \in \overline{\mathcal{R}}$, $\underline{r} \in \underline{\mathcal{R}}$
and $q'\in Q$.

We proceed by structural induction on wpa interval $\Ientff{s}{f}$ $(\text{IH}_1)$.


If $s=f$, since  $\Ientff{s}{f}$ is a wpa interval, $f$ is an internal position, 
 \begin{align*}
 (\overline{r},\underline{r},q')\in S' &\ssi \exists q\ (\overline{r},\underline{r},q) \in S \land (q, a_f, q')\in\Delta_{\text{int}} \\
 {} &\ssi \exists q\ (\overline{r},\underline{r},q)\in S \land \exec{q}{\omega, \Ientff{s}{f}}{q'}{A}
 \end{align*}
 
 Else, if there exists $i \in\Ientof{s}{f}$ such that $\Ientfo{s}{i}$ and $\Ientff{i}{f}$ are wpa, then $(\text{IH}_1)$ can be applied to both interval. Formally, by calling $S'':=L_{i-1}$ it follows :
 \begin{align*}
 (\overline{r},\underline{r},q')\in S' &\ssi \exists q''\ (\overline{r},\underline{r},q'')\in S'' \land \exec{q''}{\omega, \Ientff{i}{f}}{q'}{A} \\
 {} &\ssi \exists q, q''\ (\overline{r},\underline{r},q)\in S\land  \exec{q}{\omega, \Ientfo{s}{i}}{q''}{A} \land\exec{q''}{\omega, \Ientff{i}{f}}{q'}{A} \\
 {} &\ssi \exists q\ (\overline{r},\underline{r},q)\in S \land \exec{q}{\omega, \Ientff{s}{f}}{q'}{A}
 \end{align*}
 
 Else, since neither $s$ nor $f$ can be internals (previous case), and since $s$ (resp. $f$) cannot be a return (resp. call) of any kind, $s\in\Ientff{s}{f}^c_c$, $f\in\Ientff{s}{f}^r_r$ and $\cplii(s)=f$ (if not then $\Ientff{s}{f}$ can be split in 2 non-empty wpa intervals thus falling - again - under the previous case). 
Then, there exist $s', f'$ such that $(s, s', f', f)$ is a 2-wave thus defining the limit case for which the induction hypothesis has to be proven. 
 To conclude this case, we will prove the following property of $2$-waves:
\direct*
We defer the proof of this Lemma, and conclude that of Proposition~\ref{lemme de determinisation}
 in the limit case. We have $(\overline{r}, \underline{r}, q')\in L_f$ with $\Ientff{s}{f}$ a wpa interval and $(s, s', f', f)$ a 2-wave. Let us rename $s, s', f', f$ by $i_1, i_2, i_3, i_4$ then by construction of $\delta^r_r$, $(\overline{r}, \underline{r}, q')\in L_f$ implies the existence of $\overline{r}_1, \underline{r}_2, q, q'_1, q_2, q'_2, q_3, q'_3, q_4$ such that :
$${\psi_4}[\overline{r},\ \underline{r},\ q,\ q'\ \vert\ \overline{r}_3,\ \underline{r}_0,\ q_1,\ q'_4]$$
is satisfied. We apply lemma \ref{determinisation non wpa} giving us :
$$(\overline{r}_1, \underline{r}, q)\in L_s,\ \overline{r}_1\preceq\overline{r},\ (\overline{r}, q'_2, q_3)\in L_{i_3-1}\mbox{ and }\exec{q,q_3}{\omega, \Ientff{i_1}{i_2}, \Ientff{i_3}{i_4}}{q'_2,q'}{A}$$
Since $\Ientoo{i_2}{i_3}$ is wpa and $(\overline{r}, q'_2, q_3)\in L_{i_3-1}$ we apply $(\text{IH}_1)$, there exist $q''_2$ such that $(\overline{r}, q'_2, q''_2)\in L_{i_2}$ and $\exec{q''_2}{\omega, \Ientoo{i_2}{i_3}}{q_3}{A}$. By construction of $\delta^r_c$, we have $q''_2=q'_2$ and altogether,
$$\exec{q,q_3}{\omega, \Ientff{i_1}{i_2}, \Ientff{i_3}{i_4}}{q'_2,q'}{A}\mbox{ and }\exec{q'_2}{\omega, \Ientoo{i_2}{i_3}}{q_3}{A}$$
yields :
$$\exec{q}{\omega, \Ientff{s}{f}}{q'}{A}$$
By contradiction, if $\overline{r}_1\prec\overline{r}$ then $\overline{r}_1\in\overline{\mathcal{R}}^1$ and $\overline{r}\in\overline{\mathcal{R}}^2$ implying that $i_1-1$ and $i_3-1$ are not at the surface of the same $\cpli$-arch this would involve from the $\cpli$-arch covering $i_1-1$ to be returned in $\Ientoo{i_2}{i_3}$ and called at a position strictly before $i_1$ contradicting the wpa nature of $\Ientff{s}{f}$. Hence as $\overline{r}_1\preceq\overline{r}$ we have $\overline{r}_1=\overline{r}$ and therefore $(\overline{r}, \underline{r}, q)\in L_s$.
\end{proof}

We return to the proof of Lemma \ref{determinisation non wpa}.
To this end, we introduce some notations.
When defining the construction of the deterministic \tnwa, we introduced formula $(\phi_j)_{j\in[4]}$ with abstract names ($S_j, S'_j, \alpha_j, \dots$). In order to prove Lemma \ref{determinisation non wpa},
 we will need to introduce notations that allow us to link the construction more closely to the run $(L_i)_{i\in\Ientff{s-1}{f}}$ and a 2-wave within $\Ientff{s}{f}$. 

Let $(i_{x j})_{j\in[4]}$ be a 2-wave, where $x\in\set{\eps, \bullet, \circ}$ is a tag ($\eps$ for no tag). We derive formula $(\psi_{x m})_{m\in[4]}$ from $(\phi_m)_{m\in[4]}$ by replacing $S_j, \alpha_j$ by their counterpart $(L_{i_{x j}-1}, a_{i_{x j}})$ ($j\in [m]$) and tagging the remaining variables with $x$. Formally:
$$\psi_{x m}:={\phi_m}\left[
                        \begin{array}{l}
                          L_{i_{x j}-1},\  a_{i_{x j}},\  q_{x j},\  q'_{x j}\ \vert\ S_j,\ \alpha_j, \ q_j, \ q'_j ; \ j\in[m] \\ 
                          \overline{r}_{x 1} \ \overline{r}_{x 3} \ \underline{r}_{x 0} \ \underline{r}_{x 2} \ \vert \ \overline{r}_1, \ \overline{r}_3, \ \underline{r}_0, \ \underline{r}_2
                        \end{array}\right]$$

\begin{example} Considering 2-waves $(i_j)_{j\in[4]}$ and $(i_{\bullet j})_{j\in[4]}$: 
\begin{align*}
\psi_3 &= {\phi_3}[ L_{i_j-1},\  a_{i_j}\ \vert\ S_j,\ \alpha_j; \ j\in[3]]\\
\psi_{\bullet 4}&={\phi_4}\left[
                        \begin{array}{l}
                          L_{i_{\bullet j}-1},\  a_{i_{\bullet j}},\  q_{\bullet j},\  q'_{\bullet j}\ \vert\ S_j,\ \alpha_j, \ q_j, \ q'_j ; \ j\in[4] \\ 
                          \overline{r}_{\bullet 1}, \ \overline{r}_{\bullet 3}, \ \underline{r}_{\bullet 0}, \ \underline{r}_{\bullet 2} \ \vert \ \overline{r}_1, \ \overline{r}_3, \ \underline{r}_0, \ \underline{r}_2
                        \end{array}\right]\end{align*}
\end{example}


\begin{proof}[Proof of Lemma \ref{determinisation non wpa}]
We proceed by induction on $i_4-i_1$ $(\text{IH}_2)$. Let $k\in\Ientfo{i_1}{i_2}$, $K\in\Ientof{i_3}{i_4}$ and $p, q\in Q$ we define:
$$H(k, K, p, q):= \bigwedge\left\{\begin{array}{l}
k=i_1\land K=i_4 \mbox{ or } K\in\Ientoo{i_3}{i_4}^c_r\\
(k, K)\in\cplii \land s_{\cpli}(K-1) = i_3 \land (q'_1q'_2q'_3,p,q)\in L_{K-1}\\
\Ientof{i_1}{k}\cup\Ientfo{K}{i_4}\mbox{ is wpa and }\exec{q'_1, q}{\omega, \Ientof{i_1}{k}, \Ientfo{K}{i_4}}{p,q_4}{A}
\end{array}\right.$$
In particular $H(i_1,i_4,q'_1,q_4)$ Holds.

We state an additional technical Lemma: 
\begin{lemma} \label{lemmeHI3} If $H(k, K, p, q)$ then $\Ientoo{i_3}{K}$ is wpa or there exist $k'\in\Ientfo{i_1}{i_2}$, $K'\in\Ientoo{i_3}{K}$ and $p', q'\in Q$ such that $H(k', K', p', q')$
\end{lemma}
The purpose here is to encapsulate in one go, the way we are going to build a run on $A$ inside the $2$-wave inductively. Since we are depending on $(\text{IH}_1)$ to extract run on intervals wpa, we start from the last known linear state $K-1$, we extract the longest wpa interval finishing on $K-1$ and find $\Ientoo{i_{\bullet 4}}{K}$ (resp. $\Ientoo{i_3}{K}$) if $i_{\bullet 4}$ is the return-return position of a $2$-wave whose call-call position was covered by $i_1$ (resp. if $\Ientoo{i_3}{K}$ is wpa.) then $(\text{IH}_2)$ gives us the run for the matched $\Ientff{i_1}{i_2}$, $\Ientff{i_3}{i_4}$, and we finish applying $(\text{IH}_1)$ to $\Ientoo{i_1}{i_{\bullet 1}}$.

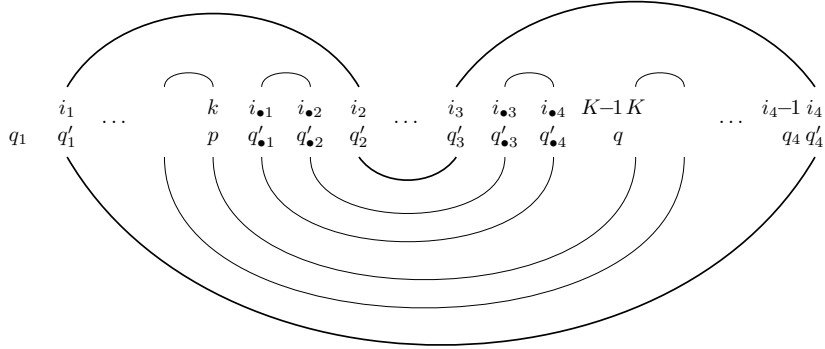
\begin{figure}[h]
$$\scalebox{0.8}{
\begin{tikzpicture}
\tikzset{node distance=0.8cm,}
\node (0) {$\begin{array}{c} \phantom{k}\\q_1 \end{array}$}; 
\node[right of = 0] (1) {$\begin{array}{c} i_1\\q'_1 \end{array}$};
\node[right of = 1] (2) {$\dots$};
\node[right of = 2] (3) {$\begin{array}{c} \phantom{k}\\ \phantom{q} \end{array}$};
\node[right of = 3] (4) {$\begin{array}{c} k \\ p \end{array}$};
\node[right of = 4] (5) {$\begin{array}{c} i_{\bullet 1}\\q'_{\bullet 1} \end{array}$};
\node[right of = 5] (6) {$\begin{array}{c} i_{\bullet 2}\\q'_{\bullet 2} \end{array}$};
\node[right of = 6] (7) {$\begin{array}{c} i_2\\q'_2 \end{array}$};
\node[right of = 7] (8) {$\dots$};
\node[right of = 8] (9) {$\begin{array}{c} i_3 \\q'_3 \end{array}$};
\node[right of = 9] (10) {$\begin{array}{c} i_{\bullet 3}\\q'_{\bullet 3} \end{array}$};
\node[right of = 10] (11) {$\begin{array}{c} i_{\bullet 4}\\q'_{\bullet 4} \end{array}$};
\node[right of = 11] (11a) {$\begin{array}{r} K\!\!-\!\!1 \\ q\end{array}$};
\node[right of = 11a, xshift = -0.25cm] (12) {$\begin{array}{c} K \\ \phantom{q}\end{array}$};
\node[right of = 12] (13) {$\begin{array}{c} \phantom{k}\\ \phantom{q} \end{array}$};
\node[right of = 13] (14) {$\dots$};
\node[right of = 14] (14a) {$\begin{array}{r} i_4\!\!-\!\!1 \\ q_4 \end{array}$};
\node[right of = 14a, xshift = -0.25cm] (15) {$\begin{array}{c} i_4 \\ q'_4 \end{array}$};

\draw[-,thick] (1.north) edge[bend left=60] (7.north);
\draw[-,thick] (9.north) edge[bend left=55] (15.north);
\draw[-] (3.north) edge[bend left=90] (4.north);
\draw[-] (5.north) edge[bend left=90] (6.north);
\draw[-] (10.north) edge[bend left=90] (11.north);
\draw[-] (12.north) edge[bend left=90] (13.north);

\draw[-,thick] (1.south) edge[bend right=60] (15.south);
\draw[-] (3.south) edge[bend right=90] (13.south);
\draw[-] (4.south) edge[bend right=90] (12.south);
\draw[-] (5.south) edge[bend right=90] (11.south);
\draw[-] (6.south) edge[bend right=90] (10.south);
\draw[-,thick] (7.south) edge[bend right=55] (9.south);
\end{tikzpicture}}$$
\caption{2-wave $(i_1, i_2, i_3, i_4)$ and notations for Lemma \ref{lemmeHI3}.}
\end{figure}

\begin{proof}[Proof of Lemma \ref{lemmeHI3}]
If $i_3=K-1$ then $\Ientoo{i_3}{K}$ is empty, hence wpa and the lemma is true. Let us suppose now that $i_3<K-1$:\\
$K''=\text{min}\set{m\in\Ientof{i_3}{K};\ \Ientoo{m}{K}\text{ is wpa}}$ is well define since $\Ientoo{K-1}{K}=\emptyset$ hence wpa, therefore $K''\leq K-1$. \\
The set $\Ientof{i_1}{k}\cup\Ientfo{K}{i_4}$ is wpa by $H(k, K, p, q)$ and $\Ientoo{K''}{K}$ is wpa  by construction, then their reunion $\Ientof{i_1}{k}\cup\Ientoo{K''}{i_4}$ is wpa. If $K''$ is a call (on any matching) it has to be returned in $\Ientoo{K''}{i_4}\subseteq\Ientof{i_1}{k}\cup\Ientoo{K''}{i_4}$ that is wpa, implying $K''\in\Ientof{i_1}{k}\cup\Ientoo{K''}{i_4}$ which is absurd. Since $K''$ being an internal would contradict its own minimality, $K''$ is a return-return position. \\
Let $k''$ such that $(k'',K'')\in \cplii$ we have $k''<i_2$ (otherwise $i_3<k''$ and $\Ientfo{k''}{K''}$ being wpa contradict minimality of $K''$).\\
Finally there exist a 2-wave $(i_{\bullet j})_{j\in[4]}$ such that $(i_{\bullet 1}, i_{\bullet 4})=(k'', K'')$.\\
As $(q'_1q'_2q'_3,p,q)\in L_{K-1}$ and $\Ientoo{K''}{K}$ wpa, $(\text{IH}_1)$ yields $q'_{\bullet 4}$ such that $(q'_1q'_2q'_3,p,q'_{\bullet 4})\in L_{K''}$ and $\exec{q'_{\bullet 4}}{\omega, \Ientoo{K''}{K}}{q}{A}$. Since $(q'_1q'_2q'_3,p,q'_{\bullet 4})\in L_{i_{\bullet 4}}$ by construction of $\delta^r_r$, we can exhibit $\overline{r}_{\bullet 1}$, $\underline{r}_{\bullet 2}$, $q_{\bullet 1}$, $q'_{\bullet 1}$, $q_{\bullet 2}$, $q'_{\bullet 2}$, $q_{\bullet 3}$, $q'_{\bullet 3}$, $q_{\bullet 4}$ such that $\psi_{\bullet 4} [q'_1 q'_2 q'_3, q'_1\vert \overline{r}_{\bullet 3}, \underline{r}_{\bullet 0}]$ is satisfied.\\
By virtue of $(\text{IH}_2)$, lemma \ref{determinisation non wpa} applies so we also have : 
$$\exec{q_{\bullet 1},q_{\bullet 3}}{\omega,  \Ientff{i_{\bullet 1}}{i_{\bullet 2}}, \Ientff{i_{\bullet 3}}{i_{\bullet 4}}}{q'_{\bullet 2},q'_{\bullet 4}}{A}$$ 
Since $r_{\bullet 1}\preceq q'_1 q'_2 q'_3$ and $r_{\bullet 1}\in \overline{\mathcal{R}}^1$ ($i_1$ is a call-call), we have $r_{\bullet 1} = q'_1$ and $(q'_1,p,q_{\bullet 1})\in L_{i_{\bullet 1}-1}$. Since $\cplii(k, K)$, $\cplii(k'', K'')$ and $\Ientoo{K''}{K}$ is wpa, $\Ientoo{k}{k''}$ is wpa. Applying $(\text{IH}_1)$ yields $p''$ such that $(q'_1,p,p'')\in L_{k}$, $\exec{p''}{\omega, \Ientoo{k}{k''}}{q_{\bullet 1}}{A}$ and by construction of $\delta^c_c$, $p=p''$\\
Altogether, we have shown:
$$\exec{p}{\omega, \Ientoo{k}{i_{\bullet 1}}}{q_{\bullet 1}}{A}\ ;\ \exec{q_{\bullet 1},q_{\bullet 3}}{\omega,  \Ientff{i_{\bullet 1}}{i_{\bullet 2}}, \Ientff{i_{\bullet 3}}{i_{\bullet 4}}}{q'_{\bullet 2},q'_{\bullet 4}}{A}\ ; \ \exec{q'_{\bullet 4}}{\omega, \Ientoo{_{\bullet 4}}{K}}{q}{A}$$
And therefore:
$$\exec{p,q_{\bullet 3}}{\omega,  \Ientof{k}{i_{\bullet 2}}, \Ientfo{i_{\bullet 3}}{K}}{q'_{\bullet 2},q}{A}$$
Which combined with: 
$$\exec{q'_1, q}{\omega, \Ientof{i_1}{k}, \Ientfo{K}{i_4}}{p,q_4}{A}$$
Yields:
$$\exec{q'_1,q_{\bullet 3}}{\omega,  \Ientof{i_1}{i_{\bullet 2}}, \Ientfo{i_{\bullet 3}}{i_4}}{q'_{\bullet 2},q_4}{A}$$
and $H(i_{\bullet 2}, i_{\bullet 3}, q'_{\bullet 2}, q_{\bullet 3})$ closes the proof of Lemma~\ref{lemmeHI3}.
\end{proof}
We are now ready to conclude the proof of Lemma~\ref{determinisation non wpa}.
Starting from $(k, K)=(i_1, i_4)$ and applying Lemma~\ref{lemmeHI3} until its hypothesis ceases to hold we exhibit $k\in\Ientfo{i_1}{i_2}$, $K\in\Ientof{i_3}{i_4}$ and $p, q\in Q$ such that $H(k, K, p, q)$ holds and $\Ientoo{i_3}{K}$ is wpa. Since $(q'_1q'_2q'_3,p,q)\in L_{K-1}$ applying $(\text{IH}_1)$ yields $q''_3\in Q$ such that $(q'_1q'_2q'_3,p,q''_3)\in L_{i_3}$ and $\exec{q''_3}{\omega, \Ientoo{i_3}{K}}{q}{A}$. By construction of $\delta^c_r$, we have $q''_3 = q'_3$.\\
As $(q'_1q'_2q'_3,p,q'_3)\in L_{i_3}$ there exists : $\overline{\tilde r}_1, \overline{\tilde r}_3, \underline{\tilde r}_0,\tilde q_1, \tilde q_2, \tilde q_3$ satisfying : 
$${\psi_3}\left[\begin{array}{l}
                     \tilde q_1,\ \tilde q_2,\ \tilde q_3\ \vert\ q_1,\ q_2,\ q_3  \\ 
                     \overline{\tilde r}_1, \ \overline{\tilde r}_{3},\ \underline{\tilde r}_0,\ p \ \vert\ \overline{r}_1,\ \overline{r}_3,\ \underline{r}_0, \ \underline{r}_2
                  \end{array}\right]$$
Granting us with : $(q'_1, p, \tilde q_2)\in L_{i_2-1}$ and $(\tilde q_2, q'_1, a_{i_2}, q'_2)\in\Delta^r_c$.\\
Since $\Ientoo{k}{i_2}$ is wpa and $(q'_1, p, \tilde q_2)\in L_{i_2-1}$ we apply $(\text{IH}_1)$, there exists $p'$ such that $(q'_1, p, p')\in L_{k}$ and $\exec{p'}{\omega, \Ientoo{k}{i_2}}{\tilde q_2}{A}$. By construction of $\delta^r_c$, we have $p=p'$, and altogether:
$$\exec{q'_3}{\omega, \Ientoo{i_3}{K}}{q}{A}\ ;\ \exec{q'_1, q}{\omega, \Ientof{i_1}{k}, \Ientfo{K}{i_4}}{p,q_4}{A}\ ;\ \exec{p}{\omega, \Ientoo{k}{i_2}}{\tilde q_2}{A}$$
Yields:
$$\exec{q'_1, q'_3}{\omega, \Ientof{i_1}{i_2}, \Ientfo{i_3}{i_4}}{\tilde q_2,q_4}{A}$$
Since $\psi_4$ gave us : $(q_1, a_{i_1}, q'_1)\in\Delta^c_c$, $(q_3, q'_2, a_{i_3}, q'_3)\in\Delta^c_r$, $(q_4, q'_3, q'_1, a_{i_4}, q'_4)\in\Delta^r_r$ and since $(\tilde q_2, q'_1, a_{i_2}, q'_2)\in\Delta^r_c$ those transitions combined with the afore mentioned ``inner'' run allows us to conclude the proof of 
Lemma~\ref{determinisation non wpa}.
\end{proof}

%% file: app-logic.tex
\section{MSO-definability of regular languages of 2-wave words}

\begin{theorem}
For all MSO-sentence $\phi$ over $(\Sigma, <, M_1, M_2)$, the set  $L(\phi)$ of all 2-wave words $(u,M_1,M_2)$ such that 
$(u,M_1,M_2) \models \phi$ is regular. 
\end{theorem}
\begin{proof}
We follow the proofs of \cite{DBLP:reference/hfl/Thomas97} and  \cite{DBLP:journals/jacm/AlurM09} for MSO over words and nested words. 
In order to simplify the proof, we consider only the restricted but equivalent  MSO$_0$-logic in where first-order variables are cancelled: we add the atomic formula $\textup{Sing}(X)$ meaning $X$ is a singleton and replace each atomic formula $A(x_1,\ldots, x_n, Y_1,\ldots, Y_m)$ by $A(X_1,\ldots, X_n, Y_1,\ldots, Y_m)$ meaning that each $X_{i}$ is a singleton $\{x_i\}$ and $A(x_1,\ldots, x_n, Y_1,\ldots, Y_m)$ holds. 

A model of a formula $\phi(X_1,\ldots, X_n)$  (that is, whose free variables are   $X_1,\ldots, X_n$) is given by a pair consisting of:
\begin{itemize}
\item a 2-wave word over the alphabet $\Sigma$ : $(a_1\cdots a_k, M_1,M_2)$, 
\item a valuation fonction, associating each $X_i$  to a subset  $I_i \subseteq [k]$, for $i=1,\ldots , n$. 
\end{itemize}
 We encode a such a model by  the 2-wave word $(v, M_1,M_2)$ over the alphabet $\Sigma \times \{0,1\}^n$ :  $v=(a_1, \beta_1) \cdots (a_k,\beta_k)$ where  for all $i\in [1,k]$, $\beta_i= b_1\cdots b_n$ and   for all $j\in [1,n]$, $b_j=1$ iff $i \in I_j$.
  
We denote by $L(\phi)$ the set of all 2-wave words satisfying $\phi$, and  show by structural induction that for all $\phi$, $L(\phi)$ is regular.  For atomic formulas $X \subseteq Y, \textup{Sing}(X), X < Y,  Q_a(X),  M_1(X,Y), M_2(X,Y)$ we proceed as for finite automata (see \cite{DBLP:reference/hfl/Thomas97}). 
For example, 
 for $M_\epsilon(X,Y)$, $\epsilon=1,2$,  it suffices to construct an automaton recognizing every 2-nested words $(v,M_1,M_2)$, $v \in \Sigma\times \{0,1\}^2$ such that there are exactly two positions $i$ and $j$ such that $\pi_2(v[i])=1$   and $\pi_3(v[j])=1$, and $(i,j)\in M_\epsilon$ (here $\pi_i$ denoted the projection on the $i$-th component, then for a letter $\alpha=(a, b_0b_1)$, $\pi_2(\alpha)$ is $b_0$ and  $\pi_3(\alpha)$ is $b_1$).

Disjunction, negation, and conjonction follow for the closure properties of regular \tnwl.  Also, if $\phi=\exists X_i \psi(X_1,\ldots X_n)$, models of $\phi$ are obtained by deleting the $i+1$-th component, that is, by applying a non erasing alphabetic morphism.  
\end{proof}

%% file: app-tw.tex


\section{Proofs of Section~\ref{sec:decision}}

\lemmatw*
\begin{proof}(sketch) Following rules of  grammar introduced for visibly 2-wave words (Lemma~\ref{lemma:twowaves}), we show by induction the following property:
\begin{itemize}
\item if $I=\interv{x_1,y_1}$ is wpa, $\omega_{|I}$ admits a tree decomposition of width $\leq 11$ containing a node $\{x_1,y_1\}$.
\item if $(I_1=\interv{x_1,y_1},I_2=\interv{x_2,y_2})$ are matched intervals, then
the graph corresponding to
$(\omega_{|I_1},\omega_{|I_2})$ 
admits a tree decomposition
of width $\leq 11$ containing 
a node $\{x_1,x_2,y_1,y_2\}$.
\end{itemize} 

Here  $\omega_{|I}$ refers to the graph induced by $\omega$ restricted to nodes in $I$. From rules of the grammars: 
\begin{itemize} 
\item either $I=(\interv{x_1,y_1},\interv{x_2,y_2})$ is wpa because  $I'=(\interv{x_1+1,y_1-1},\interv{x_2+1,y_2-1})$ is wpa and $x_1,y_1,x_2,y_2$ is a 2-wave. Then, let $t'=(V',E')$ be  the tree decomposition of $\omega_{|I'}$ given by induction hypothesis, $t=(V'\cup \{v_1,v_2\}, E' \cup \{(v_1,v_2),(v_2,v')\})$ is a tree decomposition of $\omega_{|I}$ of width less than (or equal to) 11, where: 
\begin{itemize}
\item $v_1=\{ x_1,y_1,x_2,y_2\}$ 
\item $v_2=\{ x_1,y_1,x_2,y_2,x_1+1,y_1-1,x_2+1,y_2-1\}$  (we deal here with the case where $I'$ is a pair of non-empty intervals)
\item $v'$ is the node $\{x_1+1,y_1-1,x_2+1,y_2-1\} \in V'$. 
\end{itemize}
\item Otherwise, every non trivial $I$ is constructed from  $I_1, ..., I_m$, where each $I_i$ denotes a wpa interval 
or a pair of matched intervals. From induction hypothesis each $\omega_{|I_i}$ admits a tree decomposition $t_i=(V_i,E_i)$ of width $\leq 11$. Then $t=(V'\cup \{v,v'\}, E' \cup \{(v,v')\} \cup \{(v',v_i)\}_{i\in [m]})$ is a tree decomposition of $\omega_{|I}$ of width less than (or equal to) 11, where: 
\begin{itemize}
\item $v$ is the set of all bounds of $I$
\item $v_i \in V_i$ is the set of all bounds of $V_i$, for $i\in [m]$
\item $v'$ is the union of all $v_i$, for $i\in [m]$ (it necessarily also contains $v$). 
\end{itemize}
Let us check that in all cases, $|v|,|v'| \leq 12$: 
\begin{itemize}
\item (rule $w_1w_2$) $I=\interv{x_1,y_1}$ , $I_1=\interv{x_1, y_1'}$ and $I_2=\interv{y'_1+1, y_1}$, then 
$v=\{x_1,y_1\}$ and $v'=\{x_1,y_1,y_1', y'_1+1\}$;
\item (rule $xwy$) $I=\interv{x_1,y_1}$ , $I_1=(\interv{x_1, y_1'},  \interv{x_1', y_1} )$ and $I_2=\interv{y'_1+1, x_1'-1}$, then 
$v=\{x_1,y_1\}$ and $v'=\{x_1,y_1,y_1',x_1', y'_1+1, x_1'-1\}$;
\item rule ($x_1x_2,y_1y_2$) $I$ is composed from two pairs of intervals, then $|v|=4$ and $|v' |= 8$;
\item rule ($w_1x w_1', w_2 y w_2'$) $I$ is composed of a pair of intervals and of four single intervals,  then $|v|=4$ and $|v' |= 12$.  
\end{itemize}
\end{itemize}
%
%
\end{proof}

%% file: app-discussion.tex
\section{Proofs of Section~\ref{sec:discussion}}

\paragraph*{Proof of Proposition \ref{prop:wavesToLIL}}

We prove here that removing matching relations to  regular 2-wave word languages yields linear indexed languages (\lil). 

Let us recall basic notions on Dyck words, and fix notations: 

Given an alphabet $\Gamma$, we denote by $\widebar{\Gamma}$ the  disjoint copy $\widebar \Gamma = \set{ \widebar a \mid a \in \Gamma}$.  			
	We adopt the following conventions:  $\widebar {\widebar {a}} = a$ for all $a\in \Gamma$, $\widebar{\eps}=\eps$ and  for any word  $u=\alpha_1\cdots \alpha_n \in (\Gamma \cup \widebar{\Gamma})^*$,  
	$\widebar{u}=\widebar{\alpha_n} \cdots \widebar{\alpha_1}$.  
	
\begin{definition}(Dyck words) Given an alphabet $\Gamma$, the set $\dyck_{\Gamma}$ of Dyck words over $\Gamma$ is the set of all words $u \in (\Gamma \cup \widebar{\Gamma})^*$ such that
$u \in \dyck_{\Gamma}$
\begin{itemize}
\item either $u=\epsilon$,  
\item or there exists $u_1, u_2 \in \dyck_{\Gamma}$ and $a\in \Gamma$ such that $u=au_1\widebar{a}u_2$.   	
\end{itemize}
\end{definition} 

Note that this inductive definition is non-ambiguous: if  $u=au_1\widebar{a}u_2$ and  $u=au_1'\widebar{a}u_2'$ and $u,u_1u_2,u_1',u_2' \in \dyck_{\Gamma}$), then 	$u_1=u_1'$ and $u_2=u_2'$. 
	
%
%

Dyck words are closely related to nested words:  

\begin{lemma} 
Let $u=a_1\cdots a_n \in \dyck_{\Gamma}$, and $M$ be the relation defined for all $i<j$ by $M(i,j) \Leftrightarrow  (a_i =\widebar{a_j}$ and $a_{i+1} \cdots a_{j-1} \in \dyck_{\Gamma})$. Then $(u,M)$ is a nested word whose each position is a call or a return.   
\end{lemma}

Rather than considering linear indexed grammars, we will  use an  homomorphic characterization of linear indexed languages. Such characterizations of languages are  very common, for example,  the Chomsky-Sch\"utzenberger theorem states that context-free languages are exactly languages $h(R \cap \dyck_{\Gamma})$ where $h$ is a morphism, $R$ a regular set, and $\Gamma$ an alphabet. Similarly, it is proved in \cite{DBLP:journals/iandc/FrataniV19} that indexed languages are exactly languages  $h(R \cap \dyck_{\Gamma} \cap g^{-1}(\dyck_{\Gamma'}))$, where  $R,h$  
as above, and $g$ is a \emph{symmetric} morphism, that is, $g$ is alphabetic, and  $g(\widebar{a}) = \widebar{g(a)}$. 

Remark that from a given a word $u$ in $\dyck_{\Gamma} \cap g^{-1}(\dyck_{\Gamma'})$, we can construct a $\tnw$ $(u,M_1,M_2)$ where arches in $M_2$ correspond  to $\dyck_\Gamma$, and $M_1$ to  $g^{-1}(\dyck_{\Gamma'})$.

A $\tnw$ such constructed is in fact a \ww. Regarding linear indexed languages, there exists also a homomorphic characterization of the form $h(R \cap \dyck_{\Gamma} \cap g^{-1}(\dyck_\Gamma))$. The morphism $g$ is not exactly symmetric since a counter is added to letters in order to constrain the length of waves: 

\begin{theorem}\cite{w88}
A language $L$ is linear indexed iff there exists an alphabet $A$, a regular language $R$ and a morphism $h$  such that
$$ L= h(R \cap \dyck_{\Gamma_A} \cap   g_A^{-1}(\dyck_{\Gamma_A})),$$
where : $\Gamma_A= A_1 \cup A_2$, $A_i=\{(a,i) \mid a\in A\}$, for $i=1,2$ and $g_A$ is the morphism defined by 
\begin{minipage}{0.6 \textwidth}
$$\begin{array}{lll} 
& (a,1) \mapsto  (a, 1) & \widebar{(a,1)} \mapsto  \widebar{(a,2)}\\
g_A: & & \\
& (a,2) \mapsto \widebar{(a,1)} &   \widebar{(a,2)} \mapsto  {(a,2)}\\
\end{array}$$
\end{minipage}
\begin{minipage}{0.4 \textwidth}
\scalebox{0.7}{
\begin{tikzpicture}
\tikzset{node distance=1.7cm,}
\node (1) {};
\node[right of = 1, xshift =-1cm ] (2) {$(a,1)$};
\node[right of = 2] (3) {$(a,2)$};
\node[right of = 3] (4) {$\widebar{(a,2)}$};
\node[right of = 4] (5) {$\widebar{(a,1)}$};

\node [above of = 1, yshift = -1cm] (21) {};
\node[above of = 2, yshift = -1cm] (22) {$(a,1)$};
\node[above of = 3,yshift = -1cm] (23) {$\widebar{(a,1)}$};
\node[above of = 4, yshift = -1cm] (24) {${(a,2)}$};
\node[above of = 5, yshift = -1cm] (25) {$\widebar{(a,2)}$};

 Dessin des arches
\draw[-] (2.south) edge[bend right=90] (5.south);
\draw[-] (3.south) edge[bend right =90] (4.south);
\draw[-] (22.north) edge[bend left =90] (23.north);
\draw[-] (24.north) edge[bend left =90] (25.north);
\draw[->] (1) edge node[left]{$g_A$}  (21);
\end{tikzpicture}}
\end{minipage}
\end{theorem}

%
%
%
%

To ease the proof of Proposition \ref{prop:wavesToLIL}, we will use a slightly different morphism that erases some letters.

\begin{definition}\label{def:morph}
Let $A,B$ be two disjoint alphabets, we fix  $A_i=\{(a,i) \mid a\in A\}$, for $i=1,2$, and $\Gamma_{A,B}=A_1\cup A_2 \cup B$.
Then  $g_{A,B}: \Gamma_{A,B}^* \rightarrow (A_1\cup A_2\cup \widebar{A_1}\cup \widebar{A}_2)^*$ is the morphism defined for all $a\in A$, $b\in B$  by: 
$$\begin{array}{lccccccc} 
& (a,1) &\mapsto&  (a, 1) &  & \widebar{(a,1)} &\mapsto&  \widebar{(a,2)}   \\
g_{A,B} : & (a,2) &\mapsto& \widebar{(a,1)} &  & \widebar{(a,2)} &\mapsto&  {(a,2)}  \\
& b &\mapsto& \eps & & \widebar{b} &\mapsto& \eps\\
\end{array}$$
\end{definition}

\begin{lemma}\label{lemma:weir} If a language $L$ can be written $$ L= h(R \cap \dyck_{\Gamma_{A,B}} \cap   g_{A,B}^{-1}(\dyck_{\Gamma_{A,B}})),$$
for some disjoint alphabets $A,B$, regular language $R$ and  morphism $h$,  then $L$ is a linear indexed language. 
\end{lemma}
\begin{proof}(Sketch)
One can prove that  $$g^{-1}_{A,B}(\dyck_{\Gamma_{A,B}}) \cap  \dyck_{\Gamma_{A,B}}=  
 h_{A,B}(R_{A,B} \cap g_{A\cup B}^{-1}(\dyck_{\Gamma_{A\cup B}}) \cap  \dyck_{\Gamma_{A\cup B}} ),$$ where 
 $R_{A,B} = (A_1 \cup A_2 \cup \widebar{A_1} \cup \widebar{A_2}  \cup \{(b,1)(b,2), \widebar{(b,2)}\widebar{(b,1)} \}_{b\in B})^*$, and 
$$\begin{array}{lcccccccr} 
h_{A,B} : & (a,i) &\mapsto& {(a,i)} &  & \widebar{(a,i)} &\mapsto&   \widebar{(a,i)} & \text{for $a\in A, i=1,2 $}  \\
& (b, 1) &\mapsto& b & & \widebar{(b,1)} &\mapsto& \widebar{b}& \text{for $b\in B$} \\
& (b, 2) &\mapsto& \epsilon & & \widebar{(b,2)} &\mapsto& \epsilon& \text{for $b\in B$} \\
\end{array}$$

Then  $g^{-1}_{A,B}(\dyck_{\Gamma_{A,B}}) \cap  \dyck_{\Gamma_{A,B}}$ is a \lil, and using closure properties of \lil \cite{DBLP:journals/tcs/DuskeP84}, so is $L$. 
\end{proof}

In the following we need a restricted version of \tnwa such that in any run,  all hierarchical states occurring in a cycle are equal. 

\begin{definition} A     \tnwa  $A=(Q, Q_0, Q_f, \Sigma, P, \Delta)$ is \emph{nice} if transitions satisfies: 
 $$\trans{x}{y} \subseteq \bigcup_{p\in P} Q \times \{ p\}^{\mathsf{in}_{x,y}}    \times \Sigma  \times \{ p\}^{\mathsf{out}_{x,y}}   \times Q, \ \mbox{ for all $x,y\in \{c,r\}.$}$$
\end{definition}

\begin{lemma}  
For every \tnwa $A$, there exists a nice  \tnwa $A'$ such that  $L_{ \ww_2}(A)=L_{ \ww_2}(A')$.
\end{lemma}
\begin{proof}(Sketch)  From $A=(Q, Q_0, Q_f, \Sigma, P, \Delta)$, we construct a nice  \tnwa $A'$  whose hierarchical states belongs to $P^4$, and encode all four hierarchical states labelling a  2-wave in a run of $A$. Thus, these states has to be "guessed" 
 from the call-call transition.  
 Formally, we replace each transition : 
 \begin{itemize}
\item $(q, a, p_1, p_4, q')\in \trans{c}{c}$ by $\{ (q, a, (p_1,p_2,p_3,p_4),   (p_1,p_2,p_3,p_4), q') \mid p_2,p_3 \in P\}$
\item $(q, p_1, a, p_2, q')\in \ \trans{r}{c}$  by  $\{ (q, (p_1,p_2,p_3,p_4),  a,  (p_1,p_2,p_3,p_4), q') \mid p_3,p_4 \in P\}$
\item $(q, p_2, a, p_3, q')\in \trans{c}{r}$ by  $\{ (q, (p_1,p_2,p_3,p_4),  a,  (p_1,p_2,p_3,p_4), q') \mid p_1,p_4 \in P\}$
\item  $(q, a, p_3, p_4, q')\in \trans{r}{r}$ by  $\{ (q,    (p_1,p_2,p_3,p_4) , (p_1,p_2,p_3,p_4),a, q') \mid p_1,p_2 \in P\}$.
\end{itemize}
\end{proof}

Associated with Theorem \ref{MSO-waves}, the following property proves  Proposition \ref{prop:wavesToLIL}
\begin{proposition} Let L be a regular set of $\ww_2$, the language $\existsw (L) = \{u \mid (u,M_1,M_2) \in L, \mbox{ for some }M_1,M_2\}$  is a  \lil. 
\end{proposition}

\begin{proof}
Let $A=(Q, Q_0, Q_f, \Sigma, P, \Delta)$ be a nice \nwa. 
Let  $\omega=(u,\cpli,\cplii)$ to be a 2-wave word of length $n$.  Remark that because $A$ is nice, in a given run $(\ell, h^1,h^2)$ of $A$ over $\omega$, $h^2$ is fully determined by $h^1$, because states labelling two arches having a common  extremity are equal. Then a such a run will be simply given by  $(\ell, h)$, where $h=(h_i)_{i \in [n]^{c} \cup [n]_{c} }$ is defined by  $h_i=h^1_i$ if $i \in [n]^{c}$ and  $h_i=h^2_i$ if $i \in [n]_{c}$. This definition is  consistent since $h^1_i=h^2_i$ if $i  \in [n]_{c}^{c}$.

Using notation introduced Definition \ref{def:morph}, 
we set $A=P$, $B=Q\cup \Sigma$,  and $\Gamma_{A,B} = A_1\cup A_2 \cup B$.  We define a morphism $\tau: \Delta^* \rightarrow  (\Gamma_{A,B} \cup  \widebar{\Gamma_{A,B}})^*$:  
\begin{itemize}
\item if $\delta= (q,a,q') \in \trans{ \text{int}}{ \text{int}}$ then $\tau(\delta) =\widebar{q}  a \widebar{a} q'$
\item if $\delta=(q, a, p, p, q') \in \trans{c}{c}$ then   $\tau(\delta) =\widebar{q} a \widebar{a} (p,1) q'$
\item if $\delta=(q, p, a, p, q') \in \trans{r}{c}$ then   $\tau(\delta) =\widebar{q} a \widebar{a} (p,2) q' $
\item if $\delta=(q, p, a, p, q') \in \trans{c}{r}$ then   $\tau(\delta) =\widebar{q}  a \widebar{a} \widebar{(p,2)} q' $
\item if $\delta=(q, p, p, a, q') \in \trans{r}{r}$ then   $\tau(\delta) =\widebar{q}  a \widebar{a}  \widebar{(p,1)} q' .$
\end{itemize}

For $q\in Q$, let  $R_q$  be the set words in $Q_0 \tau(\Delta^*)$  whose last letter is $q$. Clearly $R_q$  and $R=\bigcup_{q\in Q_f} R_q$ are regular languages and we consider the following two morphisms: 
\begin{itemize}
\item $f:(\Gamma_{A,B} \cup \widebar{\Gamma_{A,B}})^* \rightarrow \Sigma^*$  mapping every $a\in \Sigma$ to $a$ and all the other letters to $\varepsilon$ 
\item $g_{A,B}$ is the morphism given Definition \ref{def:morph}. 
\end{itemize}

The rest of the proof is dedicated to the the following result: $$\existsw(L_{\ww_2}(A)) =f(R \cap \dyck_{\Gamma_{A,B} } \cap g_{A,B}^{-1}(\dyck_{\Gamma_{A,B}} )).$$ According to Lemma \ref{lemma:weir}, this will prove that $\existsw(L_{\ww_2}(A))$ is a \lil. 
We split the proof in two subproofs: Lemma \ref{lem:sub1} shows that  $\existsw(L_{\ww_2}(A))  \subseteq f(R \cap \dyck_{\Gamma_{A,B} } \cap g_{A,B}^{-1}(\dyck_{\Gamma_{A,B}} ))$, and the converse inclusion follows from Lemma \ref{lem:sub2}.  

 \begin{lemma}\label{lem:sub1}
 For all $\omega=(u,M_1,M_2) \in L_{\ww_2}(A)$, there exists $w\in R \cap \dyck_{\Gamma} \cap g_{A,B}^{-1}(\dyck_\Gamma)$ such that $f(w)=u$. 
 \end{lemma}
 \begin{proof}

 Let  $\omega=(u,\cpli,\cplii)$ to be a 2-wave word of length $n$ and $s=(\ell,h)$ be an accepting  run of $A$ over $\omega$.  We associate  $s$ to the word  $w= \ell_0\tau(\delta_1 \cdots \delta_n)\widebar{\ell_n}$ where for all $i\in [n]$, $\delta_i$ is the transition applied wrt $s$. Obviously, $w \in R$ (more precisely,  $w \in R_{\ell_n}$ and ${\ell_n}\in Q_f$ because  $s$ is accepting) and $f(w)=u$.  Then, it remains to prove that $w \in \dyck_{\Gamma} \cap g_{A,B}^{-1}(\dyck_\Gamma)$.

 Let $s, f $ be indexes satisfying $1 \leq s \leq f \leq n$,  we consider the property $(H_{s,f})$: 
 \begin{enumerate}
 \item if $\interv{s,f}$ does not contain pending $M_2$ arches, then  $\ell_{s-1}\tau(\delta_{s} \cdots \delta_{f})\widebar{\ell_f}  \in \dyck_{\Gamma}$,
 \item  if $\interv{s,f}$ does not contain pending $M_1$ arches, then  $g_{A,B}(\ell_{s-1}\tau(\delta_{s} \cdots \delta_{f})\widebar{\ell_f}) \in \dyck_{\Gamma}$.
 \end{enumerate}
 In the case $s=1$, $f=n$,  $\ell_{s-1}\tau(\delta_{s} \cdots \delta_{f})\widebar{\ell_f}=w$ and  $\interv{1,n}$ is  without pending $M_1$-arches and   $M_2$-arches.  
 Then  $(H_{1,n})$ implies that $w \in \dyck_{\Gamma} \cap g_{A,B}^{-1}(\dyck_\Gamma)$.
 
Let us prove $(H_{s,f})$ by induction over $f-s$. We suppose that $u=a_1\cdots a_n$ and fix  $v=\ell_{s-1}\tau(\delta_{s} \cdots \delta_{f})\widebar{\ell_f}  $.\\ 
 (Base) If  $f=s$ and $s$ in internal (otherwise there are pending $M_1$-arches and  $M_2$-arches),  $v= \ell_{s-1}\tau(\delta_{s}) \widebar{\ell_{f}}= \ell_{s-1} \widebar{\ell_{s-1}} a_s\widebar{a_s} \ell_{f} \widebar{\ell_{f}} \in  \dyck_{\Gamma}$ and  $g_{A,B}(v)=\varepsilon \in \dyck_{\Gamma}$. \\
 (Induction) If $f>s$, we consider five cases according to $s$: 
 \begin{itemize}
 \item  if $s$ is internal, then $v=\ell_{s-1} \widebar{\ell_{s-1}} a_s\widebar{a_s} v'$ with $v' = \ell_{s} \tau(\delta_{s+1} \cdots \delta_{f})\widebar{\ell_f}$: 
  \begin{enumerate} 
 \item  if $\interv{s,f}$ does not contain pending $M_2$ arches, then  because $s$ is internal,  $\interv{s+1,f}$ does not contains pending $M_2$ arches. From hypothesis,  $v' \in \dyck_{\Gamma}$ then so is $v$. 
 \item  if $\interv{s,f}$ does not contain pending $M_1$ arches, then  because $s$ is internal,  $\interv{s+1,f}$ does not contains pending $M_1$ arches. From hypothesis,  $g_{A,B}(v') \in \dyck_{\Gamma}$, then  so is $g_{A,B}(v)=g_{A,B}(v')$. 
\end{enumerate}
\item if $s \in [n]^c_c$, there exists $i_2,i_3,i_4$ such that $s,i_2,i_3,i_4$ is a 2-wave, then there exists $p\in P$ such that: 
\begin{enumerate}
\item if  $\interv{s,f}$ does not contain pending $M_2$ arches, then $s < i_4 \leq f$, and $\interv{s+1,i_4-1}, \interv{i_4+1,f}$ do not contain pending $M_2$ arches. In addition $v$ can be decomposed in 
$$v= \ell_s \widebar{\ell_s} a\widebar{a} (p,1) v_1 a_{i_4} \widebar{a_{i_4}} \widebar{(p,1)}v_2,$$ where  $v_1=  \ell_{s+1}  \tau(\delta_{s+1} \cdots  \delta_{i_4-1})$ and  $v_2= \ell_{i_4}  \tau(\delta_{i_4+1} \cdots \delta_{f}) \widebar{\ell_f}$. \\ 
From induction hypothesis, $v_1$ and $v_2$ belong to $\dyck_{\Gamma}$, then so is 
$v$
\item   if  $\interv{s,f}$ does not contain pending $M_1$ arches, then $s < i_2 \leq f$, and $\interv{s+1,i_2-1}, \interv{i_2+1,f}$ do not contain pending $M_1$ arches. In addition $v$ can be decomposed in 
$$v= \ell_s \widebar{\ell_s} a\widebar{a} (p,1) v_1 a_{i_2} \widebar{a_{i_2}} {(p,2)}v_2,$$ where  $v_1=  \ell_{s+1}  \tau(\delta_{s+1} \cdots  \delta_{i_2-1})  \widebar{\ell_{i_2-1}}$ and $v_2= \ell_{i_2}  \tau(\delta_{i_2+1} \cdots \delta_{f}) \widebar{\ell_f}$. 
From induction hypothesis, $g_{A,B}(v_1)$ and $g_{A,B}(v_2)$ belong to $\dyck_{\Gamma}$, then so is $g_{A,B}(v)= (p,1) g_{A,B}(v_1)  \widebar{(p,1)}g_{A,B}(v_2)$. 
\end{enumerate}
\item if $s \in [n]^r_c$, then $\interv{s,f}$ contains pending $M_1$-arches. If it does not contain pending $M_2$-arches, there exists $s<i_3\leq f$ such that $M_2(s,i_3)$, $i_3 \in [n]^c_r$  and $\interv{s+1,i_3-1}, \interv{i_3+1,f}$ do not contain pending $M_2$ arches. In addition $v$ can be decomposed in 
$$v= \ell_s \widebar{\ell_s} a\widebar{a} (p,2) v_1 a_{i_3} \widebar{a_{i_3}} \widebar{(p,2)}v_2,$$ where  $v_1=  \ell_{s+1}  \tau(\delta_{s+1} \cdots  \delta_{i_3-1})  \widebar{\ell_{i_3-1}}$ and $v_2= \ell_{i_3}  \tau(\delta_{i_3+1} \cdots \delta_{f}) \widebar{\ell_f}$.  

From induction hypothesis, $v_1$ and $v_2$ belong to $\dyck_{\Gamma}$, then so is $v$. 

\item if $s \in [n]^c_r$ then $\interv{s,f}$ contains pending $M_2$-arches. If it does not contain pending $M_1$-arches, there exists $s<i_4\leq f$ such that $M_1(s,i_4)$, $i_4 \in [n]^r_r$  and $\interv{s+1,i_4-1}, \interv{i_4+1,f}$ do not contain pending $M_1$ arches. In addition $v$ can be decomposed in 
$$v= \ell_s \widebar{\ell_s} a\widebar{a} \widebar{(p,2)} v_1 a_{i_4} \widebar{a_{i_4}} \widebar{(p,1)}v_2,$$ where  $v_1=  \ell_{s+1}  \tau(\delta_{s+1} \cdots  \delta_{i_4-1})  \widebar{\ell_{i_4-1}}$ and $v_2= \ell_{i_4}  \tau(\delta_{i_4+1} \cdots \delta_{f}) \widebar{\ell_f}$.  

From induction hypothesis,  $g_{A,B}(v_1)$ and $g_{A,B}(v_2)$ belong to $\dyck_{\Gamma}$, then so is $g_{A,B}(v)= (p,2) g_{A,B}(v_1)  \widebar{(p,2)}g_{A,B}(v_2)$.  
\item  if $s \in [n]^r_r$, $\interv{s,f}$ contains pending $M_1$-arches and $M_2$-arches. 
 \end{itemize}
 \end{proof}

\begin{lemma}\label{lem:sub2}
For all $w\in R \cap \dyck_{\Gamma} \cap g_{A,B}^{-1}(\dyck_\Gamma)$, there exists a $\ww_2$ $\omega=(u,M_1,M_2)$ such that $f(w)=u$ and $u\in L(A)$.
\end{lemma}
\begin{proof}
Let $w\in R \cap \dyck_{\Gamma} \cap g_{A,B}^{-1}(\dyck_\Gamma)$. From definition of $R$, there exists $n\geq 0$ and transitions $\delta_1,\ldots, \delta_n$ such that  $w = q_0 \tau (\delta_1 \cdots \delta_n)q_f$, with $q_0\in Q_0$ and $q_f \in Q_F$.
 
For every non internal transition $\delta$, we denote by $\pi_2(\delta)$ the projection of $\tau(\delta)$ in $P_1\cup P_2 \cup \widebar{P_1} \cup \widebar{P_2}$, and by   $\pi_1(\delta)$  the word $g_{A,B}(\tau(\delta))$ (it also belongs to $(P_1\cup P_2 \cup \widebar{P_1} \cup \widebar{P_2})^*$. 
Note that because $w\in  \dyck_{\Gamma} \cap g_{A,B}^{-1}(\dyck_\Gamma)$, $ \pi_1(\delta_1\cdots \delta_n)$ and  $ \pi_2(\delta_1\cdots \delta_n)$
are Dyck words. 

 We define the $\tnw$ $\omega=(u,M_1,M_2)$ by: 
\begin{itemize}
\item $u =a_1\cdots a_n= f(w)$
\item for all $i<j$,  $M_1(i,j) \Leftrightarrow  (\delta_i$ is not internal and $ \pi_1(\delta_i) =\widebar{\pi_1(\delta_j)}$ and $\pi_1(\delta_{i+1} \cdots \delta_{j-1}) \in \dyck_{\Gamma})$. 
\item for all $i<j$,  $M_2(i,j) \Leftrightarrow  (\delta_i$ is not internal and $ \pi_2(\delta_i) =\widebar{\pi_2(\delta_j)}$ and $\pi_2(\delta_{i+1} \cdots \delta_{j-1}) \in \dyck_{\Gamma})$. 
\end{itemize}
Remark that from construction $i\in [n]^x_y$ iff $\delta_i \in [n]^x_y$. 

Now let $s=(\ell, h)$ be sequences defined  from $\omega$ and $w$ by : $\ell_0=q_0$, for all $i \in [n]$, $\ell_i$ is the last component of $\delta_i$, and for all $i\in [n]^c$, $h_i= \pi_1(\delta_i)$, for all $i\in [n]_c$, $h_i= \pi_2(\delta_i)$. 

We let the reader check that for all $i\in [n]$, $run^A_j(\omega,s)$, and   $s$ is an accepting run of $A$ over $\omega$. 
Now, it remains to prove that $\omega$ is a $\ww_2$.  From construction, for all $i\in [n]$
\begin{itemize}
\item if $i \in [n]_c^c$, there exists $p\in P$ such that $\pi_1(\delta_i)=(p,1)$ and $\pi_2(\delta_i)=(p,1)$ 
\item if $i \in [n]_c^r$, there exists $p\in P$ such that $\pi_1(\delta_i)=\widebar{(p,1)}$ and $\pi_2(\delta_i)=(p,2)$
\item if $i \in [n]_r^c$, there exists $p\in P$ such that $\pi_1(\delta_i)={(p,2)}$ and $\pi_1(\delta_i)=\widebar{(p,2)}$ 
\item  if $i \in [n]_r^r$, there exists $p\in P$ such that $\pi_1(\delta_i)=\widebar{(p,2)}$ and $\pi_1(\delta_i)=\widebar{(p,1)}$
\end{itemize}
It follows that every arch $(x,y)$ belongs to a 2-wave $i_1,i_2,i_3,i_4$ such that 
$\pi_1(\delta_{i_1}\delta_{i_2}\delta_{i_3}\delta_{i_4})=(p,1)\widebar{(p,1)}(p,2)\widebar{(p,2)}$ and  $\pi_2(\delta_{i_1}\delta_{i_2}\delta_{i_3}\delta_{i_4})=(p,1)(p,2)\widebar{(p,2)}\widebar{(p,1)}$. 
\end{proof}

\end{proof}

\paragraph*{Proof of Proposition \ref{EMOonINW}}

\EMOonINW*

\begin{proof} Consider  $\Sigma=A \cup \set{\#}$ for any alphabet $A$,  and the set $L$ of all words of the form 

$\#u_1 \# u_2\# \ldots u_{n}\#$, for $n\geq 1$, such that for all $i\in [1,n-1]$,   $u_i \in A^n$,  and if 
$u_i = a_1\ldots a_n$, then $u_{i+1}= a_n a_1 \ldots a_{n-1}$. 

Using the Shrinking Lemma given in \cite{DBLP:journals/tcs/Gilman96} for indexed languages, it can easily prove that $L$ is not an indexed language. Indeed, this Lemma ensures that for any indexed language $L$ and any positive integer $m$, there exists $k>0$, such that each word $w \in L$ with $|w|\geq k$ can be decomposed as a product of $r$  nonempty factors (with $m<r\leq k$), and such that each choice of $m$ factors  among the $r$ occurs in a   proper subproduct which lies in $L$. 

Then, choose $m=1$,  $w=\#u_1 \# u_2\# \ldots u_{n}\#$, with $n > k$, and consider a decomposition $w=w_1\cdots w_r$. As, $r \leq k$, at least one factor $w_i$ must contains two or more $\#$'s. Then $w_i$ can be decomposed in $w_i=u\#u_j\#v$ with $1\leq j \leq n$. The length of any word in $L$ containing the factor $w_i$ being the same as that of $u$, $w_i$ cannot occur in a proper subproduct of  $w_1\cdots w_r$.\\

Let us now check that $L$ is $\exists \mathit{waves}MSO$-definable. We consider the wave-nested-word language $L_{NW}$ whose word-projection is $L$, and consisting in all  $(u=\#u_1 \# u_2\# \ldots u_{n}\#, M_1,M_2)$, where $M_1$ and $M_2$ are defined as follows  (for sake of simplicity, we suppose that $n$ is odd: $n=2m$ and $m\geq 2$): 
\begin{itemize}
\item for each $u_i$, consider the unique nested word $(u_i,M_{u_i})$  such that each node in $[1,m]$ 
is a push node, and  each node in $[m+1,2m]$ is a pop node. Then  $M_1$ is union of all the $M_{u_i}$ (with the appropriate translation on node numbers); 
\item Now, for each $i\in [1,n-1]$ let us consider the unique nested word $(u_i\#u_{i+1},M'_{u_i})$  such that each node in $[m+1,2m]$ 
is a push node, and each node in $[2m+2,3m+1]$ is a pop node.    
We define $M_2$ as union of all the $M'_{u_i}$ (with the appropriate translation on node numbers)  and of  $\{(i, |u|-i), 2\leq i \leq m+1)\}$. 
\end{itemize}

Such a nested word is represented  Figure \ref{fig:cyclicword2}. 

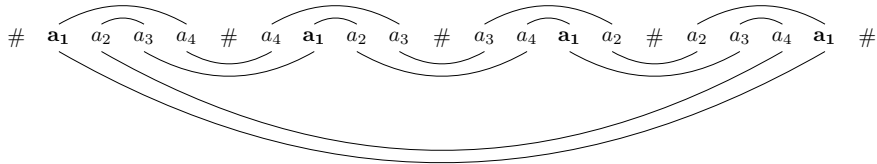
\begin{figure}[h]
\scalebox{0.8}{
\begin{tikzpicture}
\tikzset{node distance=0.7cm,}
\node (d) {$\#$};
\node[right of = d] (1) {$\mathbf{a_1}$};
\node[right of = 1] (2) {$a_2$};
\node[right of = 2] (3) {$a_3$};
\node[right of = 3] (4) {$a_4$};
\node[right of = 4] (d1) {$\#$};
\node[right of = d1] (11) {$a_4$};
\node[right of = 11] (12) {$\mathbf{a_1}$};
\node[right of = 12] (13) {$a_2$};
\node[right of = 13] (14) {$a_3$};
\node[right of = 14] (d2) {$\#$};
\node[right of = d2] (21) {$a_3$};
\node[right of = 21] (22) {$a_4$};
\node[right of = 22] (23) {$\mathbf{a_1}$};
\node[right of = 23] (24) {$a_2$};
\node[right of = 24] (d3) {$\#$};
\node[right of = d3] (31) {$a_2$};
\node[right of = 31] (32) {$a_3$};
\node[right of = 32] (33) {$a_4$};
\node[right of = 33] (34) {$\mathbf{a_1}$};
\node[right of = 34] (d4) {$\#$};
\draw[-] (1.north) edge[bend left =30] (4.north);
\draw[-] (2.north) edge[bend left =30] (3.north);

\draw[-] (11.north) edge[bend left =30] (14.north);
\draw[-] (12.north) edge[bend left =30] (13.north);

\draw[-] (21.north) edge[bend left =30] (24.north);
\draw[-] (22.north) edge[bend left =30] (23.north);

\draw[-] (31.north) edge[bend left =30] (34.north);
\draw[-] (32.north) edge[bend left =30] (33.north);


\draw[-] (1.south) edge[bend right =30] (34.south);
\draw[-] (2.south) edge[bend right =30] (33.south);
\draw[-] (3.south) edge[bend right =30] (12.south);
\draw[-] (4.south) edge[bend right =30] (11.south);
\draw[-] (13.south) edge[bend right =30] (22.south);
\draw[-] (14.south) edge[bend right =30] (21.south);
\draw[-] (23.south) edge[bend right =30] (32.south);
\draw[-] (24.south) edge[bend right =30] (31.south);

\end{tikzpicture}}
\caption{Example of a word in $L$}
\label{fig:cyclicword2}
\end{figure}

First we construct a formula FO formula  defining the set of all wave-NW $(\#u_1\#... \#u_k\#, M_1,M_2)$  such that $k\geq 2$,  there exists $m\geq 2$ such that  $|u_i|=2m$, for $1\leq i \leq k$, and $M_1,M_2$ are defined as above.  Now we can ensure  the property  that if 
$u_i = a_1\ldots a_n$, then $u_{i+1}= a_n a_1 \ldots a_{n-1}$,  since the labeling has to satisfy the following local properties: if $M_1(x,y) \wedge M_2(y,z) \wedge M_1(z,t)$ then
\begin{itemize}
\item  $u(x)=u(z+1)$ and  $u(y-1)=u(t)$
\item if $u(y+1)=\#$, then $u(y)=u(z)$
\end{itemize}
Until now, we have constructed a FO formula. Now, it remains to check that $k=2m$.  For this purpose, we follow the  path of the first letter of the first factor, and check that it reaches the end of the word (then $k$ is a multiple of $2m$), and it is the only one time it reaches the end of a factor (then $k< 4m$).  To follows  the path  of this letter,  we define the set $X$  of all elements $x$ satisfying one of the following property:
  \begin{itemize} 
\item $x$ is the first position of $u_1$; 
\item  there exist $y \in X$, and $z$  such that $(y,z) \in M_1$ and $(z,x-1) \in M_2$;
\item  there exist $y \in X$, and $z$  such that $(y,z) \in M_2$ and $(z,x-1) \in M_1$. 
\end{itemize}  If $X$ contains the last position of the factor $u_k$ and it is the only  one $x \in X$ indexing the last letter of a factor $u_i$, for $1 \leq i \leq k$, then $k=2m$.   
Symbols indexed by elements  of  $X$ are represented in bold  in  Figure \ref{fig:cyclicword}. 
\end{proof}